\documentclass[aps,prd,twocolumn,superscriptaddress,nofootinbib,showpacs,longbibliography]{revtex4-1}
\usepackage{amsthm,amsmath,amssymb}
\usepackage{float}

\newlength{\bracewidth}
\newcommand{\myunderbrace}[2]{\settowidth{\bracewidth}{$#1$}#1\hspace*{-1\bracewidth}\smash{\underbrace{\makebox{\phantom{$#1$}}}_{#2}}}
\newcommand*\mystrut[1]{\vrule width0pt height0pt depth#1\relax}

\newcommand{\ee}{\mathrm{e}}
\newcommand{\ii}{\mathrm{i}}
\newcommand{\eig}{\mathcal{E}}

\newtheorem{theorem}{Theorem}

\newtheorem{corollary}{Corollary}
\newtheorem*{conjecture}{Conjecture}

\usepackage{dsfont}  
\usepackage{enumitem} 

\usepackage{arydshln} 
\usepackage{tikz,tikz-3dplot}
\usetikzlibrary{external}
\definecolor{lgrey}{gray}{0.9}

\usepackage[plainpages=false,pdfpagelabels,colorlinks=true,linkcolor=red,urlcolor=blue,citecolor=blue,pdftitle={Entanglement production in bosonic systems: Linear and logarithmic growth},pdfauthor={},pdfdisplaydoctitle=true,pdfduplex=DuplexFlipLongEdge]{hyperref}

\begin{document}

\title{Entanglement production in bosonic systems: Linear and logarithmic growth}
\author{Lucas Hackl}
\affiliation{Institute for Gravitation and the Cosmos, The Pennsylvania State University, University Park, PA 16802, USA}
\affiliation{Department of Physics, The Pennsylvania State University, University Park, PA 16802, USA}
\author{Eugenio Bianchi}
\affiliation{Institute for Gravitation and the Cosmos, The Pennsylvania State University, University Park, PA 16802, USA}
\affiliation{Department of Physics, The Pennsylvania State University, University Park, PA 16802, USA}
\author{Ranjan Modak}
\affiliation{Department of Physics, The Pennsylvania State University, University Park, PA 16802, USA}
\author{Marcos Rigol}
\affiliation{Department of Physics, The Pennsylvania State University, University Park, PA 16802, USA}

\begin{abstract}
We study the time evolution of the entanglement entropy in bosonic systems with time-independent, or time-periodic, Hamiltonians. In the first part, we focus on quadratic Hamiltonians and Gaussian initial states. We show that all quadratic Hamiltonians can be decomposed into three parts: (a) unstable, (b) stable, and (c) metastable. If present, each part contributes in a characteristic way to the time-dependence of the entanglement entropy: (a) linear production, (b) bounded oscillations, and (c) logarithmic production. In the second part, we use numerical calculations to go beyond Gaussian states and quadratic Hamiltonians. We provide numerical evidence for the conjecture that entanglement production through quadratic Hamiltonians has the same asymptotic behavior for non-Gaussian initial states as for Gaussian ones. Moreover, even for non-quadratic Hamiltonians, we find a similar behavior at intermediate times. Our results are of relevance to understanding entanglement production for quantum fields in dynamical backgrounds and ultracold atoms in optical lattices.
\end{abstract}

\maketitle

\section{Introduction}
Entanglement production has been extensively studied in physical systems ranging from quantum fields and gravity \cite{Calabrese:2005in, calabrese2007quantum, cotler2016entanglement, Hubeny:2007xt, AbajoArrastia:2010yt, Hartman:2013qma, Liu:2013qca, Bianchi:2014bma} to condensed matter \cite{dechiara2006entanglement, fagotti2008evolution, eisler2008entanglement, lauchli2008spreading, kim2013ballistic, alba2017entanglement} and quantum information \cite{giovannetti2014ultimate,gagatsos2016entropy,de2017gaussian}. It has been recently probed experimentally in systems of ultracold bosonic atoms in optical lattices \cite{islam_ma_15, kaufman_tai_16}. In this paper, we develop methods that allow one to compute the linear and logarithmic contributions to the entanglement entropy production for the most general bosonic quadratic Hamiltonian that is time independent, which includes Floquet Hamiltonians for periodically-driven systems. In particular we prove that, for any subsystem $A$, the time evolution of the entanglement entropy $S_A(t)$ shows the characteristic behavior
\begin{align}\label{eq:1}
    S_A(t)=\Lambda_A\,t+C_A\,\ln(t)+X_A(t)\,,
\end{align}
where $\Lambda_A$ is a real number, $C_A$ is an integer, and $X_A$ is a bounded function. The leading order term $\Lambda_A\,t$ agrees with previous results on the connection between entanglement growth and Lyapunov exponents in unstable systems \cite{bianchi17}. The subleading logarithmic term $C_A\ln(t)$ is a feature that appears in quadratic Hamiltonians that are metastable. We show how to compute $\Lambda_A$ and $C_A$ directly from the Hamiltonian, and investigate cases with $\Lambda_A=0$, for which the entanglement entropy grows logarithmically.

Previous studies of entanglement production in bosonic systems have focused on quantum quenches involving stable Hamiltonians. In finite systems (not in a many-body localized phase \cite{nandkishore_huse_review_15}), a regime of linear growth of the entanglement entropy is necessarily followed by saturation \cite{dechiara2006entanglement}. In the integrable case, in which a quasi-particle picture is available, the linear growth can be understood in terms of propagation of quasi-particles \cite{Calabrese:2005in}. The saturation, on the other hand, is the result of equilibration. The systems are usually prepared in some initial state $|\psi_0\rangle$, with expectation value of the energy $E_0=\langle\psi_0|\hat{H}|\psi_0\rangle$, and then are let evolve unitarily, i.e., $|\psi_t\rangle=e^{-i\hat{H}t}|\psi_0\rangle$. For a \emph{stable} local Hamiltonian $\hat{H}$, the Gibbs state has the maximum entropy at energy $E_0$, i.e., this entropy, which is extensive, bounds from above the entanglement entropy that the system can reach when it equilibrates. Saturation to the entropy of the Gibbs state occurs only in quantum chaotic systems, with integrable ones saturating at a smaller (extensive) value \cite{dalessio_kafri_16}. Here, we study a different type of Hamiltonians, namely, unstable or metastable ones. We study quenches in which the system is prepared in a pure state $|\psi_0\rangle$ and is let evolve unitarily under such Hamiltonians. In this case, the entanglement entropy can grow without bound.

Instabilities in bosonic systems appear in a variety of forms. The simplest example is perhaps the one of the inverted harmonic oscillator. In quantum field theory in dynamical backgrounds, instabilities give rise to a wealth of non-equilibrium processes in which knowledge of the entanglement dynamics is of phenomenological relevance. In cosmological inflation, momentum modes $(\vec{k},-\vec{k})$ of the quantum fluctuations of the metric and the inflaton field become unstable when they cross the Hubble radius \cite{mukhanov2005physical,weinberg2008cosmology}. The resulting amplification of perturbations provides the quantum seeds for the temperature inhomogeneities in the cosmic microwave background, and the study of the dynamics of the entanglement entropy in this process has been proposed as a tool for investigating the cosmological quantum-to-classical transition \cite{Campo:2005sy,Polarski:1995jg,Kiefer:1999sj,Martin:2015qta}. On the other hand, momentum modes of quantum fields that are within the Hubble radius can become unstable via the phenomenon of parametric resonance \cite{traschen1990particle,kofman1994reheating,allahverdi2010reheating,amin2015nonperturbative}. This process results in a large non-thermal production of particles called pre-heating. Once the produced particles thermalize, they provide the initial conditions for the hot big-bang phase of the primordial universe. 

Similar mechanisms have been proposed for the so-called ``little bang,'' the production of a quark-gluon plasma in heavy-ion collisions \cite{yagi2005quark}. The study of entanglement production for these systems is expected to provide new insights into the mechanism of pre-thermalization, as well as new tools for estimating the relevant time scales involved in the process  \cite{muller2011entropy,kunihiro2010chaotic,Hashimoto:2016wme}. A third example of bosonic systems in which instabilities lead to a rich phenomenology is the case of ultracold atomic gases trapped in an external potential that is periodically modulated \cite{fedichev2004cosmological,carusotto2010density}. This modulation can induce a response in the Bose-Einstein condensate that leads to stimulated quasi-particle production. Studying the dynamics of the entanglement entropy in these systems is of particular relevance because of current experiments that can probe the non-separability of quasi-particle pair creation \cite{jaskula2012acoustic,Steinhauer:2015saa}.

A comprehensive discussion of the linear growth of the entanglement entropy in field theoretical systems with unstable modes can be found Ref.~\cite{bianchi17}, in which Gaussian states and linear symplectic methods were employed \cite{holevo2013quantum,weedbrook2012gaussian,adesso2014continuous,Bianchi:2015fra,Bianchi2017kahler}. Here we extend this analysis by determining the subleading logarithmic corrections to the entanglement production, and by studying numerically the dynamics of non-Gaussian states and under non-quadratic (interacting) Hamiltonian evolution. The numerical results allow us to put forward a conjecture that widens significantly the realm of relevance of our analytical results. While our methods are tailored to applications to quantum fields in dynamical backgrounds, our presentation focuses on quantum mechanical systems with a finite set of bosonic modes, which can be understood as a multi-mode generalization of the two-mode squeezing of the $(\vec{k},-\vec{k})$ sector.\footnote{This generalization is relevant for instance for wavepacket observables \cite{Bianchi2017entropy}, and for quantum fields evolving in a background that is not necessarily homogeneous and isotropic.}

While Gaussian states play a prominent role in the analytic description of a variety of physical systems, they are only an approximation. Ever present interactions result in non-Gaussian states, and, even for initial Gaussian states, weak interactions can produce strongly non-Gaussian states over long times. It is therefore important to put to a test the robustness of our analytical results for non-Gaussian states and non-quadratic Hamiltonians. This is done numerically within a two-site Bose-Hubbard model in the limit in which interactions are very weak. In the absence of interactions, this model was studied analytically in Ref.~\cite{ghosh2017entanglement} by solving the non-linear Ermakov equation. Here, we present an analytic solution that relies on linear symplectic methods, and confirms our theoretical predictions for the linear, logarithmic, and oscillatory contributions to the evolution of the entanglement entropy. For non-Gaussian initial states, our analytical results for Gaussian states (which act as upper bounds for the numerical ones) are close to being saturated.

We should add that, in recent years, entanglement production after quantum quenches in fermionic and spin systems has also drawn much attention~\cite{kim2013ballistic, khlebnikov_kruczenski_14, alba14, alba_17, alba2017entanglement}. While both integrable and quantum chaotic systems exhibit an intermediate linear growth of the entanglement entropy, a new class of systems, many-body localized systems, has emerged in which the growth of the entanglement entropy is much slower, namely, logarithmic~\cite{prosen.2008, bardarson.2012, serbyn.2013}. Many-body localized systems are striking because they are interacting systems that are robust against eigenstate thermalization~\cite{deutsch1991, srednicki1994, rigol.2008, dalessio_kafri_16}, and as such they do not thermalize when taken away from equilibrium~\cite{khatami_rigol_12a, serbyn_Papic_14, tang2015quantum, nandkishore_huse_review_15}. The logarithmic growth of the entanglement entropy after a quench is considered to be another hallmark of the many-body localized phase, which differentiates it from the (noninteracting) Anderson localized one \cite{anderson.1958} (both are generated by disorder and exhibit no dc conductivity). It is remarkable that, in the analytical and numerical studies of the bosonic systems considered here, the growth of the entanglement entropy is only found to be either linear or logarithmic, as in quantum quenches in many-body quantum systems in delocalized and many-body localized phases, respectively. 

The presentation is structured as follows. In Sec.~\ref{BosonicSystems}, we review basic definitions for bosonic systems and set the notation used in the subsequent parts. In Sec.~\ref{theoretical}, we use analytical techniques to derive the asymptotic time dependence of the entanglement entropy for time-independent quadratic Hamiltonians and Gaussian initial states. Using computational methods, in Sec.~\ref{numerical} we explore dynamics involving non-Gaussian initial states and non-quadratic Hamiltonians. In Sec.~\ref{discussion}, we make some general remarks about the mechanism for entanglement entropy production studied in this work. We discuss applications of our results, and put forward a conjecture on the entanglement entropy of non-Gaussian initial states evolving under quadratic Hamiltonians. In the appendixes, we provide relevant supplements on the Jordan normal form, the classical time evolution, and on our numerical calculations.

\section{Bosonic systems\label{BosonicSystems}}
We consider bosonic systems with $N$ degrees of freedom. We can fix a basis of bosonic creation and annihilation operators $\hat{a}^\dagger_i$ and $\hat{a}_j$ satisfying the commutation relations
\begin{align}
	[\hat{a}_i,\hat{a}_j]=[\hat{a}^\dagger_i,\hat{a}^\dagger_j]=0,\quad\text{and }\quad [\hat{a}_i,\hat{a}_j^\dagger]=\delta_{ij}\,.
\end{align}
Quantum states can be described as square integrable complex functions on $\mathbb{R}^N$ or as elements of the Fock space generated from a vacuum state $|0\rangle$ with $\hat a_i|0\rangle=0$ for all $i$. The latter approach turns out to be more useful and we define our Fock space as
\begin{align}
	\begin{split}
	&\mathcal{H}_V=\mathrm{span}\left\{|n_1,\cdots,n_N\rangle\Big|n_i\in\mathbb{N}\right\},\\
	&\text{with}\quad|n_1,\cdots,n_N\rangle=\left(\prod_{i=1}^{N}\frac{(\hat{a}^\dagger_i)^{n_i}}{\sqrt{n_i!}}\right)|0\rangle\,,
	\end{split}
\end{align}
where the states $|n_1,\cdots,n_N\rangle$ form an orthonormal basis for $\mathcal{H}_V$. Let us emphasize that this basis of states and, in particular, the Fock space vacuum $|0\rangle$, are, in general, completely independent from eigenstates and the ground state of the Hamiltonian (to be chosen later). We only use this basis to parametrize states in the Hilbert space and to decompose the system into different subsystems.

\subsection{Quadratic Hamiltonians\label{sec:tihamilt}}
In the first part of this paper, we focus on Hamiltonians that are quadratic in terms of creation and annihilation operators. The most general quadratic Hamiltonian is
\begin{align}
	\hat{H}=\frac{1}{2}\sum^N_{i,j=1}\left[\Delta_{ij}\hat{a}_i^\dagger \hat{a}_j^\dagger+\Delta^*_{ij}\hat{a}_i\hat{a}_j+\gamma_{ij}(\hat{a}_i^\dagger \hat{a}_j+ \hat{a}_j\hat{a}_i^\dagger)\right]\,,
\end{align}
where the matrices $\Delta$ and $\gamma$ satisfy
\begin{align}
	\Delta^\intercal=\Delta\,,\quad\gamma^\intercal=\gamma^*\,.
\end{align}
The expressions and computations simplify if we switch from creation and annihilation operators to the Hermitian operators
\begin{align}\label{new basis}
	\hat{q}_i=\frac{1}{\sqrt{2}}(\hat{a}_i^\dagger+\hat{a}_i)\,,\quad
	\hat{p}_i=\frac{\ii}{\sqrt{2}}(\hat{a}_i^\dagger-\hat{a}_i)\,.
\end{align}
We can choose a basis $\hat{\xi}^a\equiv(\hat{q}_1,\cdots,\hat{q}_N,\hat{p}_1,\cdots,\hat{p}_N)$ and rewrite $\hat H$ as
\begin{align}
	\hat{H}=\frac{1}{2}h_{ab}\hat{\xi}^a\hat{\xi}^b,\quad\text{with}\quad h\equiv\left[\begin{array}{cc}
	\mathrm{Re}(\Delta+\gamma) & \mathrm{Im}(\Delta+\gamma)\\[.0em]
	\mathrm{Im}(\Delta-\gamma) & \mathrm{Re}(\gamma-\Delta)
	\end{array}\right]\,.
\end{align}
Here, $h$ is an arbitrary symmetric matrix that contains exactly the same amount of information as the matrices $\Delta$ and $\gamma$ together. We use Einstein's summation convention, i.e., we sum over contracted indices. A contracted index only refers to a pair of a lower index and an upper index. In the Hamiltonian $\hat H=\frac{1}{2}h_{ab}\hat{\xi}^a\hat{\xi}^b$, the indices $a$ and $b$ are contracted. If the reader is familiar with Penrose's abstract index notation, contracted indices can also be read as contracted in Penrose's sense. All such equations are valid independent of the basis that one chooses to write them in. However, we also give explicit expressions for the relevant matrices, such as $h$, with respect to the basis of choice $\hat{\xi}^a\equiv(\hat{q}_1,\cdots,\hat{q}_N,\hat{p}_1,\cdots,\hat{p}_N)$. Here, we use the symbol ``$\equiv$'' to emphasize that the expression is only valid with respect to this basis. For instance, the commutation relations in terms of $\hat{\xi}^a$ can be written as
\begin{align}
	[\hat{\xi}^a,\hat{\xi}^b]=\ii\Omega^{ab},\quad\text{with}\quad \Omega\equiv\left(\begin{array}{cc}
	0 & \mathds{1}\\[0em]
	-\mathds{1} & 0
	\end{array}\right)\,.
\end{align}
The inverse of $\Omega^{ab}$ is given by $\omega_{ab}$, such that $\omega_{ab}\Omega^{bc}=\delta_a{}^c$ and $\Omega^{ab}\omega_{bc}=\delta^a{}_c$.

In many situations, one diagonalizes quadratic Hamiltonians by finding a Bogoliubov transformation into eigenmodes. This corresponds to a Bogoliubov transformation
\begin{align}
	\hat{c}_k=\sum^N_{k,i=1}(\alpha_{ki}\hat{a}_i+\beta_{ki}\hat{a}^\dagger_i)\,,
\end{align}
such that the Hamiltonian takes the simple form
\begin{align}
	\hat{H}=E_0+\frac{1}{2}\sum^N_{k=1}\epsilon_k\,\hat{c}_k^\dagger \hat{c}_k\,.
\end{align}
However, such a transformation does not exist in general. This fact is directly related to the spectrum and decomposition of the matrix
\begin{align}
	K^a{}_b=\Omega^{ac}h_{cb}\,,
\end{align}
called the symplectic generator of classical time evolution. Only if $h$ is positive definite, the Hamiltonian $\hat{H}=\frac{1}{2}h_{ab}\hat{\xi}^a\hat{\xi}^b$ is positive definite and can be diagonalized through a Bogoliubov transformation. In this case, $K$ is diagonalizable with purely imaginary eigenvalues.

Let us consider examples for a single degree of freedom with creation and annihilation operator $\hat a^\dagger$ and $\hat a$. We also use the operators $\hat{q}=\frac{1}{\sqrt{2}}(\hat{a}^\dagger+\hat{a})$ and $\hat{p}=\frac{i}{\sqrt{2}}(\hat{a}^\dagger-\hat{a})$.
\begin{enumerate}
	\item[(a)] \textbf{Inverted harmonic oscillator (unstable)}\\
	The first example is the quantum version of an inverted harmonic oscillator with upside-down potential $V(\hat{q})=-\hat{q}^2$:
	\begin{align}
		\begin{split}
		&\qquad\hat{H}=-\frac{1}{2}\left[(\hat{a}^\dagger)^2 +\hat{a}^2\right]=\frac{1}{2}\left(\hat{p}^2-\hat{q}^2\right),\quad\text{with}\\
		&\qquad h\equiv\left(\begin{array}{cc}
		-1 & 0\\[0em]
		0 & 1
		\end{array}\right)\,\Rightarrow\, K\equiv\left(\begin{array}{cc}
		0 & 1\\[0em]
		1 & 0
		\end{array}\right)\,.
		\end{split}
	\end{align}
	This Hamiltonian is not bounded from below and it does not have eigenstates in the Fock space. However, we can still evolve arbitrary states with it. We can think of $\hat{H}$ as the quadratic expansion of a quartic Hamiltonian with potential $V(\hat{q})=-\hat{q}^2+\epsilon\hat{q}^4$, which is bounded from below and has a regular spectrum. The eigenvalues of $K$ are given by $\pm 1$.
	\item[(b)] \textbf{Harmonic oscillator (stable)}\\
	The second example is the well-known quantum harmonic oscillator:
	\begin{align}
		\begin{split}
		&\quad \hat{H}=\frac{1}{2}\left(\hat{a}^\dagger \hat{a}+\hat{a} \hat{a}^\dagger\right)=\frac{1}{2}\left(\hat{p}^2+\hat{q}^2\right),\quad\text{with}\\
		&\quad h\equiv\left(\begin{array}{cc}
		1 & 0\\[0em]
		0 & 1
		\end{array}\right)\,\Rightarrow\, K\equiv\left(\begin{array}{cc}
		0 & 1\\[0em]
		-1 & 0
		\end{array}\right)\,.
		\end{split}
	\end{align}
	It is already written in terms of normal modes with $h$ being positive definite. The eigenvalues of $K$ are $\pm i$ and the Hamiltonian is bounded from below. Moreover, it is diagonalizable with a complete basis of eigenstates that span the Hilbert space. We refer to this Hamiltonian as stable because, classically, it corresponds to a potential with a single global minimum.
	\item[(c)] \textbf{Free Hamiltonian (metastable)}\\
	The third example is the Hamiltonian of a free particle in one dimension:
	\begin{align}
		\begin{split}
		&\hat{H}=-\frac{1}{2}\left(\hat{a}-\hat{a}^\dagger\right)^2=\frac{1}{2}\hat{p}^2,\quad\text{with}\\
		&h\equiv\left(\begin{array}{cc}
		1 & 0\\[0em]
		0 & 0
		\end{array}\right)\,\Rightarrow\, K\equiv\left(\begin{array}{cc}
		0 & 0\\[0em]
		-1 & 0
		\end{array}\right)\,.
		\end{split}
	\end{align}
	This Hamiltonian is bounded from below but its eigenvectors are plane waves, which cannot be normalized, leading to a continuous spectrum. We refer to this Hamiltonian as metastable because, classically, it corresponds to a flat potential which does not have a single global minimum.
\end{enumerate}
Note that we cannot study entanglement production in these simple examples because all three systems consist of only a single degree of freedom. We study composite systems with many degrees of freedom that show features of all three examples above, and study the entanglement for different system decompositions.

\subsection{Floquet Hamiltonian}
The quadratic Hamiltonians in Sec.~\ref{sec:tihamilt} are all time-independent, but the same methods also apply to time-periodic Hamiltonians. A quadratic Hamiltonian with time dependence given by
\begin{align}
	\hat{H}(t)=\frac{1}{2}h(t)_{ab}\hat{\xi}^a\hat{\xi}^b
\end{align}
is time periodic if $h(t)_{ab}=h(t+T)_{ab}$ for some period $T$. Such Hamiltonians describe periodically driven systems. Interestingly, this does not imply that the time-evolution is periodic nor that the entanglement entropy just oscillates. We can write the time evolution operator as the time-ordered exponential
\begin{align}
	\hat U(t)=\mathcal{T}\exp\left[-i\int^t_0dt'\hat{H}(t')\right]\,.
\end{align}
At stroboscopic times, $t=nT$ with $n\in\mathbb{N}$, we can write $\hat U(t)=U(T)^n$. We can therefore compute the stroboscopic time evolution of such systems using the time-independent Floquet Hamiltonian
\begin{align}
	\hat{H}_{\mathrm{F}}=\frac{1}{T}\ln \hat U(T)\,.
\end{align}
Provided that $\hat{H}(t)$ is quadratic, the Floquet Hamiltonian $\hat{H}_{\mathrm{F}}$ is quadratic as well. This means that our methods for time-independent Hamiltonians can still be used to study periodically driven systems.

\subsection{Subsystems}
We are interested in computing the entanglement entropy for different decompositions
\begin{align}
	\mathcal{H}_V=\mathcal{H}_{A}\otimes\mathcal{H}_{B}
\end{align}
of two subsystems $\mathcal{H}_{A}$ and $\mathcal{H}_{B}$ with $N_A$ and $N_B$ degrees of freedom, respectively (with $N=N_A+N_B$). Without loss of generality, we choose our subsystems to be the ones generated by the first (last) $N_A$ ($N_B$) creation operators. This means that
\begin{align}
	\mathcal{H}_{A}&=\mathrm{span}\left\{|n_1,\cdots,n_{N_A},0,\cdots,0\rangle \Big|n_i\in\mathbb{N}\right\}\,,\\[0em]
	\mathcal{H}_{B}&=\mathrm{span}\left\{|0,\cdots,0,n_{N_A+1},\cdots,n_{N}\rangle\Big|n_i\in\mathbb{N}\right\}\,.
\end{align}
Given an arbitrary state $|\psi\rangle\in\mathcal{H}_V$, its entanglement entropy associated subsystem $A$ is
\begin{align}
	S_A(|\psi\rangle)=-\mathrm{tr}_{\mathcal{H}_A}
	\hat\rho_A\ln \hat\rho_A\,,\quad\hat\rho_A=\mathrm{tr}_{\mathcal{H}_B}|\psi\rangle\langle\psi|\,,
\end{align}
where we first trace over the degrees of freedom in $\mathcal{H}_B$, to obtain the reduced state $\hat\rho_A$, and then compute its von Neumann entropy. Another important measure of correlations is the Renyi entropy (of order $2$)
\begin{align}
	R_A(|\psi\rangle)=-\ln\left(\mathrm{tr}_{\mathcal{H}_A}\hat\rho_A^2\right)\,.
\end{align}

It is useful to understand subsystems in terms of a subset of observables that probe the relevant portion of the Hilbert space. In our case, all observables probing $\mathcal{H}_A$ are built from the linear observables $\hat{\xi}^a_A\equiv(\hat{q}_1,\cdots,\hat{q}_{N_A},\hat{p}_1,\cdots,\hat{p}_{N_A})$. We use this to restrict the $2$-point functions $G^{ab}$ to subsystem $A$. We then write $[G]_A^{ab}$ for the restricted $2$-point function (see next section). In the above basis, this restriction can be understood as selecting the $2N_A\times 2N_A$ sub-matrix corresponding to the correlations of the linear observables $\hat{\xi}^a_A$ in the subsystem. More mathematically, we can understand a choice of subsystem as a phase space decomposition $V=A\oplus B$, where $A$ and $B$ are complementary symplectic\footnote{A symplectic sub-vector is a sub-space for which the restricted symplectic form is still non-degenerate. Equivalently, it means that we can find a basis of canonically conjugate variables.} sub-vector spaces. Equivalently, we can think of this as a splitting of linear observables into those observables that only probe subsystem $A$ and those that only probe subsystem $B$, which implies a decomposition of the dual phase space $V^*=A^*\oplus B^*$. Using this terminology, $[G]^{ab}_A$ is the restriction of the bilinear form $G^{ab}$ to the sub-vector space $A$. Similarly, we are able to restrict a linear map $J: V\to V$ to the subsystem $A$, such that $[J]_A: A\to A: a\mapsto \mathrm{pr}_A\left[J(a)\right]$ acts on a vector $a\in A$, by first applying $J$ and then projecting onto $A$.

\subsection{Gaussian states}
For our analytic calculations of entanglement entropy production under quadratic Hamiltonians, we restrict ourselves to the class of Gaussian states. Computing the entanglement entropy of an arbitrary state with respect to arbitrary system decompositions is very difficult because the bosonic Hilbert space is infinite dimensional. However, for Gaussian states we can use powerful tools enabling us to compute the entanglement entropy directly from the 2-point correlation function.

We can define a Gaussian state in terms of its $n$-point correlation function. Given an arbitrary state $|\psi\rangle$, we define its $1$-point and $2$-point functions as
\begin{align}
	z^a&=\langle\psi|\hat{\xi}^a|\psi\rangle\,,\\
	G^{ab}&=\langle\psi|(\hat{\xi}^a-z^a)(\hat{\xi}^b-z^b)+(\hat{\xi}^b-z^b)(\hat{\xi}^a-z^a)|\psi\rangle\,.
\end{align}
In the basis $\hat{\xi}^a\equiv(\hat{q}_1,\cdots,\hat{q}_N,\hat{p}_1,\cdots,\hat{p}_N)$, $z^a$ is a $2N$-dimensional vector and $G^{ab}$ is a $2N\times 2N$ matrix. A general connected $n$-point function is then defined as
\begin{align}
	C^{a_1\cdots a_n}_{|\psi\rangle}=\langle\psi|\,\mathrm{Sym}\!\left[\left(\hat{\xi}^{a_1}-z^{a_1}\right)\cdots\left(\hat{\xi}^{a_n}-z^{a_n}\right)\right]|\psi\rangle\,,
\end{align}
where $\mathrm{Sym}$ denotes full symmetrization over all indices. We refer to a state as Gaussian if $n$-point functions for odd $n$ vanish and if all connected $2n$-point functions can be computed from $G^{ab}$ via Wick's theorem, namely
\begin{align}
	C^{a_1\cdots a_{2n}}_{|\psi\rangle}=\sum_{\sigma}G^{a_{\sigma(1)}a_{\sigma(2)}}\cdots G^{a_{\sigma(2n-1)}a_{\sigma(2n)}}\,,
\end{align}
where $\sigma$ goes through all permutations satisfying $\sigma(2i)>\sigma(2i-1)$ for all $i$. One can use $G$ and $z$ to characterize a Gaussian state $|G,z\rangle$ uniquely, and can compute the entanglement entropy of Gaussian states directly from the $2N\times 2N$ matrix $G$. Note that $G^{ab}$ is a positive-definite inner product on the dual phase space $V^*$.

Given a state $|\psi\rangle$ that is not Gaussian, we can always use its $1$-point function $z^a$ and $2$-point function $G^{ab}$ to define a Gaussian state $\varrho_{|\psi\rangle}$ with
\begin{align}
z^a&=\mathrm{Tr}\left(\hat{\xi}^a\varrho_{|\psi\rangle}\right)\,,\\
\quad \frac{1}{2}(G^{ab}+i\Omega^{ab})&=\mathrm{Tr}\left((\hat{\xi}^a-z^a)(\hat{\xi}^b-z^b)\varrho_{|\psi\rangle}\right)\,,
\end{align}
by requiring that $n$-point functions for odd $n$ vanish and for even $n$ can be computed from Wick's theorem. A subtlety lies in the fact that $\varrho_{|\psi\rangle}$ is in general not a pure state, even though $|\psi\rangle$ is pure. We can call $\varrho_{|\psi\rangle}$ the Gaussian part of $|\psi\rangle$, and its entanglement entropy bounds the entanglement entropy of $|\psi\rangle$ from above \cite{bianchi17}.

\subsection{Entanglement entropy}
The entanglement entropy $S_A(|G,z\rangle)$ associated with a subsystem $A$ of the Gaussian state $|G,z\rangle$ is completely encoded in its $2$-point function $G^{ab}$. An elegant computation method is provided by the linear complex structure $J$ defined as the matrix $J^a{}_b=-G^{ac}\omega_{cb}$. A linear complex structure provides an equivalent characterization of the state $|G,z\rangle$ because its (complex) eigenspace with eigenvalue $-i$ can be interpreted as all possible linear combinations of annihilation operators that annihilate $|G,z\rangle$. Restricting $J$ to subsystem $A$ gives rise to the so called restricted complex structure $[J]_A$, which is the sub-matrix containing only entries with respect to a basis of $A\subset V$. The eigenvalues of $[J]_A$ are purely imaginary, such that $[\ii J]_A$ has real eigenvalue pairs $\pm \nu_i$ with $\nu_i$ being the same as the symplectic eigenvalues of the restricted covariance matrix $[G]_A$. The entanglement entropy of Gaussian states is usually \cite{Adesso:2007tx} computed as
\begin{align}
\begin{split}
&S_A(|G,z\rangle)=\sum^N_{i=1}S(\nu_i),\qquad\text{with}\qquad\quad\\
&S(\nu_i)=\frac{\nu_i+1}{2}\ln\frac{\nu_i+1}{2}-\frac{\nu_i-1}{2}\ln\frac{\nu_i-1}{2}\,,
\end{split}
\end{align}
which can be reformulated into a simple trace formula \cite{bianchi17} in terms of $[\ii J]_A$ given by
\begin{align}
	S_A(|G,z\rangle)=\mathrm{Tr}\left(\frac{\mathds{1}_A+[\ii J]_A}{2}\right)\ln\left|\frac{\mathds{1}_A+[\ii J]_A}{2}\right|\,.\label{eq:master}
\end{align}
Here, the trace is just over $2N_A$-dimensional matrices. For a highly entangled system, the entanglement entropy $S_A$ approaches the Renyi entropy (of order $2$) $R_A$. For a Gaussian state $|G,z\rangle$, the Renyi entropy is given by \cite{bianchi17}
\begin{align}
	R_A(|G,z\rangle)=\frac{1}{2}\ln|\det[\ii J]_A|=\ln\mathrm{Vol}(\mathcal{V}_A)\,.
\end{align}
The equality $\frac{1}{2}\ln|\det[\ii J]_A|=\ln\mathrm{Vol}(\mathcal{V}_A)$ can be derived in the following way: First, we express $J$ with respect to the basis $(q_1,\cdots,q_{N},p_1,\cdots,p_{N})$, such that the equation $\det[\ii J]_A=\det[-\ii G\omega]_A$ simplifies to $\det[G]_A$ (with respect to this basis). Second, a determinant of $[G]_A$ with respect to a basis corresponds to the volume of the parallelepiped $\mathcal{V}_A$ spanned by this basis. The volume is measured by the volume form on $A^*$ induced by the inner product $[G]_A$. This means that we can use $[G]_A$ to compute length and angles between the basis vectors $(q_1,\cdots,q_{N_A},p_1,\cdots,p_{N_A})$, and this leads to the volume $\mathrm{Vol}(\mathcal{V}_A)$ appearing in the formula.

The Renyi entropy of Gaussian states bounds the entanglement entropy from both sides through
\begin{align}
	R_A(|G,z\rangle)\leq S_A(|G,z\rangle)\leq R_A(|G,z\rangle)+(\ln{2}-1)N_A\,.\label{Renyi-inequality}
\end{align}
The left inequality applies to any state. It is a well-known relation between the Renyi entropy of order $2$ and the von Neumann-entropy. The right inequality only applies to Gaussian states and was derived in Ref.~\cite{bianchi17}. This inequality is the reason why we can study the asymptotics of $S_A$ by analyzing $R_A$ for $R_A\to\infty$. In this limit, the inequality above implies $\lim_{R_A\to\infty}S_A/R_A=1$. We use this result to compute the asymptotic time evolution of $S_A(t)$ in terms of $R_A(t)$.

For non-Gaussian states $|\psi\rangle$, we cannot use the formulas above to compute the entanglement entropy $S_A(|\psi\rangle)$ or the Renyi entropy $R_A(|\psi\rangle)$. However, we can still compute $z^a$ and $G^{ab}$ from $|\psi\rangle$ and apply the formulas above to compute the entanglement entropy of the Gaussian state $\varrho_{|\psi\rangle}$. Importantly, the entanglement entropy $S_A(|\psi\rangle)$ is bounded from above by the Gaussian entanglement entropy $S_A(\varrho_{|\psi\rangle})$, as explained in Refs.~\cite{holevo1999capacity,bianchi17}:
\begin{align}
	S_A(|\psi\rangle)\leq S_A(\varrho_{|\psi\rangle})= \mathrm{Tr}\left(\frac{\mathds{1}_A+[\ii J]_A}{2}\right)\ln\left|\frac{\mathds{1}_A+[\ii J]_A}{2}\right|.
\end{align}
This statement can be phrased as: \emph{Among all states with the same $2$-point function $G^{ab}$, the Gaussian state has the maximal entanglement entropy.} Note, however, that the Gaussian state with covariance matrix $G^{ab}$ is not necessarily a pure state.

\subsection{Time evolution}
We focus on the dynamics generated by time-independent quadratic Hamiltonians $H=\frac{1}{2}h_{ab}\xi^a\xi^b$ (the time-dependent case can be treated in a similar fashion~\cite{bianchi17}). For quadratic Hamiltonians, the classical Hamilton equations of motion are linear and take the following form:
\begin{align}
	\dot{\xi}^a=\{H,\xi^a\}=\Omega^{ab}\partial_bH=\Omega^{ab}h_{bc}\xi^b\,,
\end{align}
where a vector $\xi^a\in V$ has components $\xi^a\equiv(q_1,\cdots,q_N,p_1,\cdots,p_N)$ with respect to a canonical basis. The solution $\xi^a(t)$ of this equation is encoded in the classical Hamiltonian flow
\begin{align}
	M(t)^a{}_b=\exp(tK)^a{}_b,\quad\text{with}\quad K^a{}_b=\Omega^{ac}h_{cb}\,,
\end{align}
such that $\xi^a(t)=M(t)^a{}_b\xi^b_0$ for some initial condition $\xi_0^a$. The matrix $K$ is called the (symplectic) generator of the time evolution. We derive a general theorem for entanglement production by analyzing the properties of $M(t)$ and the generator $K$.

It is crucial to recall that, for quadratic Hamiltonians, the quantum evolution of expectation values is completely encoded in the classical evolution, commonly known as Ehrenfest's theorem. More precisely, given an arbitrary initial state $|\psi_0\rangle$, the time-dependent $n$-point function for the state $|\psi(t)\rangle=e^{-it\hat{H}}|\psi_0\rangle$ is given by
\begin{align}
	C_{|\psi(t)\rangle}^{a_1\cdots a_n}=M(t)^{a_1}{}_{b_1}\cdots M(t)^{a_t}{}_{b_n}C_{|\psi_0\rangle}^{b_1\cdots b_n}\,.\label{correlation}
\end{align}
In particular, we have the $2$-point function $G(t)^{ab}=M(t)^a{}_cM(t)^b{}_d G_0^{cd}=M(t) G_0 M^\intercal(t)$. Here, $M^\intercal(t)_b{}^a=M(t)^a{}_b$ refers to the transpose of $M(t)$, which is important if we write matrix products rather than showing the index contraction explicitly. In Eq.~(\ref{correlation}), we do not assume that the state is Gaussian.

For the special case of a Gaussian initial state $|\psi_0\rangle=|G_0,z_0\rangle$, the state remains Gaussian and its parameter $G$ and $z$ can be computed from $M(t)$. The solution of Schr\"odinger's equation is therefore given by
\begin{align}
	|\psi(t)\rangle=|M(t)G_0M^\intercal(t),M(t)z_0\rangle\,.
\end{align}
This equation allows us to rewrite the Renyi entropy associated with $|\psi(t)\rangle$ as the time-dependent quantity
\begin{align}
	R_A(t)=\ln\mathrm{Vol}[M^\intercal(t)\mathcal{V}_A]\,.
\end{align}
This is due to the fact that measuring the volume of the time-independent region $\mathcal{V}_A$ with respect to the time-dependent metric $M(t)G_0M^\intercal(t)$ is equivalent to measuring the time-dependent volume $M^\intercal(t)\mathcal{V}_A=\{ M^\intercal(t)_b{}^a\theta_a|\theta_a\in \mathcal{V}_A\subset A^*\}$ with respect to the time-independent metric $G_0$. Provided that the Renyi entropy grows asymptotically without bound, the entanglement entropy grows with the same asymptotics
\begin{align}
	S_A(t)\sim R_A(t)=\ln\mathrm{Vol}[M^\intercal(t)\mathcal{V}_A]\,,
\end{align}
which follows from the inequality in Eq.~(\ref{Renyi-inequality}).

\section{Entanglement production of quadratic Hamiltonians\label{theoretical}}
In this section, we study the evolution of the entanglement entropy for Gaussian states under quadratic time-independent Hamiltonians.

\subsection{Decomposition of quadratic Hamiltonians}
Every quadratic Hamiltonian $\hat{H}=\frac{1}{2}h_{ab}\hat{\xi}^a\hat{\xi}^b$ can be uniquely decomposed into the three parts
\begin{align}
\hat{H}=\hat{H}_\text{unstable}+\hat{H}_\text{stable}+\hat{H}_\text{metastable}\,.
\end{align}
We show that these three parts contribute to the time-evolution of the entanglement entropy $S_A(t)$ in a characteristic way, namely:
\begin{enumerate}
	\item[(a)] Unstable Hamiltonian\\
	$\Rightarrow$ $S_A(t)$ entropy grows linearly: $S_A(t)\sim \Lambda_A t$
	\item[(b)] Stable Hamiltonian\\
	$\Rightarrow$ $S_A(t)$ oscillates: $S_A(t)\sim X_A(t)$
	\item[(c)] Metastable Hamiltonian\\
	$\Rightarrow$ $S_A(t)$ grows logarithmically: $S_A(t)\sim C_A\ln{(t)}$
\end{enumerate}
The decomposition is best understood by looking at the matrix $K^a{}_b=\Omega^{ac}h_{cb}$. This is a real square matrix and, as such, can always be decomposed into the following three commuting parts
\begin{align}
K=K_{\text{real}}+K_{\text{imaginary}}+K_{\text{nilpotent}}\,,
\end{align}
such that $K_{\text{real}}$ is a diagonalizable matrix with real eigenvalues, $K_{\text{imaginary}}$ is a diagonalizable matrix with imaginary eigenvalues, and $K_{\text{nilpotent}}$ is a nilpotent matrix. This decomposition is the well-known Jordan decomposition of real matrices, which we review in Appendix~\ref{jordan}. In order to find the decomposition of $\hat{H}$, we contract the different parts of $K$ with $\omega_{ab}$:
\begin{align}
\hat{H}_{\text{unstable}}&=\frac{1}{2}(h_\text{unstable})_{ab}\hat{\xi}^a\hat{\xi}^b\,,\\ 
\hat{H}_{\text{stable}}&=\frac{1}{2}(h_\text{stable})_{ab}\hat{\xi}^a\hat{\xi}^b\,,\\ 
\hat{H}_{\text{metastable}}&=\frac{1}{2}(h_\text{metastable})_{ab}\hat{\xi}^a\hat{\xi}^b\,,
\end{align}
where $(h_\text{unstable})_{ab}=\omega_{ac}(K_\text{real})^c{}_b$, $(h_\text{stable})_{ab}=\omega_{ac}(K_\text{imaginary})^c{}_b$, and $(h_\text{metastable})_{ab}=\omega_{ac}(K_\text{nilpotent})^c{}_b$. The condition for two quadratic Hamiltonians to commute is given by
\begin{align}
[\hat{H}_1,\hat{H}_2]=\frac{1}{2}\underbrace{\left[(h_1)_{ab}\Omega^{bc}(h_2)_{cd}-(h_2)_{ab}\Omega^{bc}(h_1)_{cd}\right]}_{=\omega_{ab}[K_1,K_2]^b{}_d}\hat{\xi}^a\hat{\xi}^d=0\,.
\end{align}
This condition is equivalent to $[K_1,K_2]=K_1K_2-K_2K_1=0$, namely, requiring that the corresponding matrices $(K_i)^a{}_b=\Omega^{ac}(h_i)_{cb}$ commute. We can conclude that the decomposition of a Hamiltonian into the three aforementioned parts induces an equivalent decomposition of the time evolution operator into the three commuting parts
\begin{align}
\hat U(t)=\ee^{-\ii\hat{H}_\mathrm{unstable}t}\ee^{-\ii\hat{H}_\mathrm{stable}t}\ee^{-\ii\hat{H}_\mathrm{metastable}t}\,,
\end{align}
where each part contributes to the time dependence of the entanglement entropy.

The asymptotics of the entanglement entropy is closely related to how the classical Hamiltonian flow $M(t)=\ee^{Kt}$ deforms regions of the classical phase space [see Appendix~\ref{classical} for a derivation of how $M(t)$ corresponds to the classical flow solving the classical Hamiltonian equations of motion].

We are mostly interested in the dual flow $M^\intercal(t)=\ee^{tK^\intercal}$ on the dual phase space that describes the time evolution of classical observables. This flow can be best understood by studying its action on a linear observable $\theta\in V^*$ in the dual phase space. To do this, we first bring $K^\intercal$ into its Jordan normal form
\begin{align}
\begin{split}
&K^\intercal\equiv\left(\begin{array}{cccc}
\boxed{J(\kappa_1)} & & &\\[0em]
& \boxed{J(\kappa_2)} & & \\[0em]
& &\ddots &\\[0em]
& & & \boxed{J(\kappa_n)}
\end{array}\right)\,,\\
&\quad\text{with}\quad J(\kappa)\equiv\left(\begin{array}{ccc}
\boxed{J_1(\kappa)} & &\\[0em]
&\ddots &\\[0em]
& & \boxed{J_{j_\kappa}(\kappa)}
\end{array}\right)\,,
\end{split}
\end{align}
by finding a basis consisting of a complete set of generalized eigenvectors for every Jordan block $J_k(\kappa)$ associated with the eigenvalue $\kappa$:
\begin{itemize}
	\item \textbf{Real eigenvalue $\kappa=\lambda$:}\\
		For every Jordan block $J_k(\kappa)$ of the real eigenvalue $\kappa$, we have $\dim J_k(\kappa)$ distinct generalized eigenvectors $\eig^l_k(\kappa)$. Here, $\kappa$ runs over all real eigenvalues of $K$, $k$ runs over the number of Jordan blocks associated with $\kappa$, and $l$ runs up to the dimension $\dim J_k(\kappa)$. The action of $M^\intercal(t)$ on $\eig^l_k(\kappa)$ is given by
		\begin{align}
			M^\intercal(t)\,\eig^l_k(\kappa)=\ee^{\lambda t}\sum^{l}_{l'=1}\frac{t^{l-l'}}{(l-l')!}\eig^{l'}_k(\kappa)\,.
		\end{align}
		Its length behaves asymptotically as $\ln\lVert M^\intercal(t)\eig^{l'\pm}_k(\kappa)\rVert\sim\lambda t+ (l-1)\ln(t)$ as $t\to\infty$.
	\item \textbf{Complex eigenvalue $\kappa=\lambda+\ii \omega$}\\
	For every Jordan block $J_k(\kappa)$ of the complex eigenvalue $\kappa$ with $\omega>0$, all generalized eigenvectors $\eig^{l\pm}_k(\kappa)$ come in pairs. Therefore, we have the additional label $\pm$ to distinguish the two vectors per pair, besides the labels $\kappa$, $k$, and $l$. The action of $M^\intercal(t)$ on $\eig^l_k(\kappa)$ is given by
	\begin{eqnarray}
	&&\qquad\ M^\intercal(t)\,\eig^{l+}_k(\kappa)=\\&&\qquad\ \ee^{\lambda t}\sum^{l}_{l'=1}\frac{t^{l-l'}}{(l-l')!}\left[\cos(\omega t)\eig^{l'+}_k(\kappa)+\sin(\omega t)\eig^{l'-}_k(\kappa)\right],\nonumber\\
	&&\qquad\ M^\intercal(t)\,\eig^{l-}_k(\kappa)=\\&&\qquad\ \ee^{\lambda t}\sum^{l}_{l'=1}\frac{t^{l-l'}}{(l-l')!}\left[\cos(\omega t)\eig^{l'-}_k(\kappa)-\sin(\omega t)\eig^{l'+}_k(\kappa)\right],\nonumber
	\end{eqnarray}
	Its length behaves asymptotically as $\ln\lVert M^\intercal(t)\eig^{l'\pm}_k(\kappa)\rVert\sim\lambda t+ (l-1)\ln(t)$ as $t\to\infty$, which is the same as in the real case.
\end{itemize}
The decomposition discussed here is closely related to the normal forms of quadratic Hamiltonians discussed in Ref.~\cite{arnold1990symplectic}.

\subsection{Asymptotic volume growth}
In order to study the asymptotic behavior of the entanglement entropy, we need to understand how the volume of the parallelepipeds $\mathcal{V}$ grows asymptotically under the action of the Hamiltonian flow $M^\intercal(t)$. In order to answer this question, it is helpful to first consider $L$ vectors $\Phi^j$ selected out of the generalized eigenvectors, namely
\begin{align}
\Phi^j=\eig^{l_j\sigma_j}_{k_j}(\kappa_{i_j})\,,
\end{align}
where $\sigma_j$ vanishes for real $\kappa_j$, but can take values $\sigma_j\in\{+,-\}$ for complex $\kappa_j$. Given such a set, we can define the $L$-dimensional parallelepiped spanned by them as the set
\begin{align}
	\mathcal{V}=\left\{\sum^{L}_{j=1}c_j\Phi^j\Big|0\leq c_j\leq 1\right\}\,,
\end{align}
which is a subset of the hyperplane $\mathrm{span}(\Phi^1,\cdots,\Phi^L)\subset V^*$. We define the time evolved parallelepiped as
\begin{align}
	M^\intercal(t)\mathcal{V}=\left\{\sum^{L}_{j=1}c_j\,M^\intercal(t)\theta^j\Big|0\leq c_j\leq 1\right\}\,,
\end{align}
where each spanning vector $\Phi^j$ evolves to $M^\intercal(t)\Phi^j$. The following theorem provides a precise answer for how the volume grows asymptotically.

\begin{theorem}[Volume asymptotics\label{volume}]
The asymptotic behavior of the volume $M^\intercal(t)\mathcal{V}$ is
\begin{align}
&\ln\mathrm{Vol}[M^\intercal(t)\mathcal{V}]\sim \Lambda\,t+C\,\ln(t)\,,\quad\text{with}\\
&\Lambda=\sum^L_{j=1}\lambda_{i_j}\quad{\rm and}\quad C=\sum^L_{j=1}l_j-\sum_{n^{\sigma}_k(\kappa)}\frac{n^{\sigma}_k(\kappa)[n^{\sigma}_k(\kappa)+1]}{2}\,,\nonumber
\end{align}
where $n^{\sigma}_k(\kappa)$ is the number of vectors $\theta^j$ that were selected out of the $m_k(\kappa)$ vectors
\begin{align}
	\eig^{1\sigma}_k(\kappa)\,,\quad\eig^{2\sigma}_k(\kappa)\,,\quad\cdots\quad\eig^{j_k(\kappa)\sigma}_k(\kappa)\,.
\end{align}
For real $\kappa$, $\sigma$ vanishes, and for complex $\kappa$, we have $\sigma\in \{+,-\}$. Note that this asymptotic behavior is universal, that is, it is independent of the specific (time-independent) metric or volume form used to compute it.
\end{theorem}

\begin{proof}
We know that the length $\lVert M^\intercal(t)\eig^{l_j\sigma_{j}}_{k_j}(\kappa_{i_j})\rVert$ of each vector grows asymptotically as
\begin{align}
\ln\lVert M^\intercal(t)\eig^{l_j\sigma_{j}}_{k_j}(\kappa_{i_j})\rVert\sim \lambda_{i_j}t+(l_j-1)\ln(t)\,.
\end{align}
From this, we can make a first guess that the asymptotic volume growth should be given by
\begin{align}
\sum^L_{j=1}\left[\lambda_{i_j}t+(l_j-1)\ln(t)\right]\,.
\end{align}
However, one can convince oneself that the second term cannot be correct if two or more vectors come from the same sequence consisting of the $m_k(\kappa)$ vectors
\begin{align}
\eig^{1\sigma}_k(\kappa)\,,\quad\eig^{2\sigma}_k(\kappa)\,,\quad\cdots\,,\quad\eig^{j_k(\kappa)\sigma}_k(\kappa)\,.
\end{align}
Let $n^\sigma_k(\kappa)$ be the number of vectors in this sequence that we selected as part of the $L$ vectors $\Phi^j$. Let us refer to these vectors as
\begin{align}
	\eig^{r_j\sigma}_k(\kappa)
\end{align}
with $1\leq r_j\leq n^\sigma_k(\kappa)$ and $r_j<r_{j+1}$. When we evolve each of these vectors, their dominating growth points into the same direction, namely
\begin{align}
	\eig^{r_j\sigma}_k(\kappa)\sim t^{r_j-1}\,M^\intercal(t)\eig^{1\sigma}_k(\kappa)\,.
\end{align}
Even though each vector grows with the asymptotic $\ee^{\lambda t} t^{r_j-1}$, the volume spanned by these vectors cannot be the product of the length growth because all $n^\sigma_k(\kappa)$ grow dominantly in the same direction, the direction that $M^\intercal(t)\eig^{1\sigma}_k(\kappa)$ is evolving. In order to find the volume growth, we need to consider the first $n_k^\sigma(\kappa)$ linearly independent directions that these vectors are dominantly growing into. The dominant directions will therefore be the first $n^\sigma_k(\kappa)$ vectors of the above sequence, namely
\begin{align}
\eig^{1\sigma}_k(\kappa)\,,\quad\eig^{2\sigma}_k(\kappa)\,,\quad\cdots\quad\eig^{n^\sigma_k(\kappa)\sigma}_k(\kappa)\,.
\end{align}
This means the $j$-th vector of our sequence only contributes $(r_j-j)$ rather than $r_j-1$ as power of $t$ to the overall volume growth. If we sum the exponents for all $j$, we find
\begin{align}
	\sum^{n_k^\sigma(\kappa)}_{j=1}(r_j-j)=\sum^{n_k^\sigma(\kappa)}_{j=1}r_j-\frac{n_k^\sigma(\kappa)[n_k^\sigma(\kappa)+1]}{2}\,.
\end{align}
This is the contribution for the vectors $\theta^j$ belonging to a specific sequence. If we sum over all contributions from vectors $\theta^j$ belonging to all possible sequences, we find
\begin{align}
&\ln\mathrm{Vol}[M^\intercal(t)\mathcal{V}]\sim \Lambda\,t+C\,\ln(t),\quad\text{with}\\
&\Lambda=\sum^L_{j=1}\lambda_{i_j}\quad{\rm and}\quad C=\sum^L_{j=1}l_j-\sum_{n^{\sigma}_k(\kappa)}\frac{n^{\sigma}_k(\kappa)[n^{\sigma}_k(\kappa)+1]}{2},\nonumber
\end{align}
as expected. Let us point out two subtleties that are crucial for the argument:
\begin{itemize}
	\item For complex eigenvalues, the time evolution does not just stretch vectors $\eig^{l\sigma}_k(\kappa)$, but also rotates them. However, this does not change their asymptotic growth, just an overall prefactor that may depend on time as the vector rotates in the subspace spanned by $\eig^{l\sigma}_k(\kappa)$ for $1\leq l\leq \dim J_k(\kappa)/2$ and $\sigma\in\{+,-\}$. This prefactor is always bounded because the rotation occurs in a closed orbit with frequency $\omega=\mathrm{Im}(\kappa)$. Finally, we do not need to worry about the fact that vectors in the $\sigma\!=\!+$ sequence become linear dependent on vectors in the $\sigma\!=\!-$ sequence. Even though they rotate in the same subspace, they always stay linearly independent and do not approach the same direction due to a phase difference of $\pi/2$.
	\item A similar argument can be used to explain why we can consider sequences associated with different Jordan blocks $J_k(\kappa)$ independently and do not need to worry about different vectors approaching the same direction. The time evolution of vectors $\eig^{l\sigma}_k(\kappa)$ with fixed $k$ and $\kappa$ always stay in the subspace spanned by them and never approach directions of vectors with different $k'$ or $\kappa'$.
\end{itemize}
This concludes the proof.
\end{proof}
We found the precise volume asymptotics for a given choice of $L$ vectors $\Phi^j$. The next question is for which choice of $L$ vectors $\Phi^j$ the volume grows most rapidly. This is answered by the next theorem.
\begin{theorem}[Maximal volume growth\label{maxvolum}]
The following algorithm allows us to find $L$ vectors $\Phi^j\in V^*$, such that the corresponding parallelepiped $\mathcal{V}\subset V^*$ grows most rapidly among all $L$-dimensional parallelepipeds, and its asymptotics is given by
\begin{align}
	\ln M^\intercal(t)\mathcal{V}\sim \Lambda^L_{\max}t+C^L_{\max}\ln{(t)}\,,
\end{align}
where $\Lambda_{\max}$ and $C_{\max}$ are computed below.
\begin{enumerate}
	\item We define the exponential and polynomial contribution of a vector $\eig^{l\sigma}_k(\kappa)$ as
	\begin{align}
		&\mathbb{E}\left(\eig^{l\sigma}_k(\kappa)\right)=\mathrm{Re}(\kappa)\quad\text{and}\\
		&\mathbb{P}\left(\eig^{l\sigma}_k(\kappa)\right)=\left\{\begin{array}{lcl}
		2l-1-\dim J_k(\kappa),& &\mathrm{Im}(\kappa)=0\\[0em]
		2l-1-\dim J_k(\kappa)/2, & &\mathrm{Im}(\kappa)\neq0
		\end{array}\right.\,.\nonumber
	\end{align}
	\item We sort all $2N$ generalized eigenvectors $\eig^{l\sigma}_k(\kappa)$ into a long list
	\begin{align}
		\left(\Phi^1,\Phi^2,\cdots,\Phi^{2N}\right)\,,
	\end{align}
	such that $\mathbb{E}(\Phi^j)\geq\mathbb{E}(\Phi^{j+1})$ is always satisfied and such that $\mathbb{P}(\Phi^j)\geq\mathbb{P}(\Phi^{j+1})$ is satisfied whenever $\mathbb{E}(\Phi^j)=\mathbb{E}(\Phi^{j+1})$. This sorting may not be unique, but it is sufficient for finding the maximal volume growth.
	\item A maximally growing parallelepiped is spanned by the first $L$ vectors $\Phi^j$ and its asymptotics is given by
	\begin{align}
	\qquad \Lambda^L_{\max}=\sum^L_{j=1}\mathbb{E}(\Phi^j)\quad\text{and}\quad C^L_{\max}=\sum^L_{j=1}\mathbb{P}(\Phi^j)\,.
	\end{align}
	Note that we maximize the asymptotics and not $\Lambda$ and $C$ individually.
\end{enumerate}
\end{theorem}
\begin{proof}
The proof goes in two steps. First, we ask what specific vector $\eig^l_k(\kappa)$ contributes to the asymptotics of the volume, and second, how do we need to sort them in order to get maximal contributions.\\
\textbf{Step 1}: Every vector $\eig^{l\omega}_k(\kappa)$ may contribute to both the exponential asymptotics $\Lambda$ and the polynomial asymptotics $C$. From Theorem \ref{volume}, we recall that a specific vector $\eig^{l\sigma}_k(\kappa)$ contributes
\begin{align}
	\mathbb{E}(\eig^{l\sigma}_k(\kappa))=\lambda=\mathrm{Re}(\kappa)
\end{align}
to the exponential asymptotics. For the contribution to the polynomial asymptotics, we need to know how many vectors of the same sequence are already contributing. If there are already $s$ vectors in the sequence, the vector $\eig^{l\sigma}_k(\kappa)$ contribute exactly 
\begin{align}
	\mathbb{P}(\eig^{l\sigma}_k(\kappa))=(l-1)-s
\end{align}
to the polynomial exponent. However, if there are $m_k(\kappa)$ vectors in the sequence, we would only choose the vector $\eig^l_k(\kappa)$ for our parallelepiped if we have already chosen the $s=m_k(\kappa)-l$ vectors $\eig^{l+1}_k(\kappa),\cdots,\eig^{m_k(\kappa)}_k(\kappa)$. Here, we have $m_k(\kappa)=\dim J_k(\kappa)$ for real eigenvalue $\kappa$ and $m_k(\kappa)=\dim J_k(\kappa)/2$ for complex eigenvalue $\kappa$. In total, this leads to a contribution of
\begin{align}
	\mathbb{P}\left(\eig^{l\sigma}_k(\kappa)\right)=\left\{\begin{array}{lcl}
	2l-1-\dim J_k(\kappa),& &\mathrm{Im}(\kappa)=0\\[0em]
	2l-1-\dim J_k(\kappa)/2, & &\mathrm{Im}(\kappa)\neq0
	\end{array}\right.\,.
\end{align}
\textbf{Step 2:} It is clear that when we choose vectors of a parallelepiped in order, we maximize the asymptotics of its volume if we choose vectors first based on their exponential contribution and only second based on their polynomial contribution. Moreover, if two vectors have identical exponential and polynomial contribution, it does not matter which one we choose.
\end{proof}

\subsection{Asymptotic entanglement production}
When studying the time evolution of the entanglement entropy, the following volume formula for Gaussian states is of much help:
\begin{align}
	S_A(t)\sim \ln\mathrm{Vol}\!\left[M^\intercal(t)\mathcal{V}_A\right]\,,
\end{align}
where $\mathcal{V}_A\subset A$ is an arbitrary $2N_A$-dimensional parallelepiped in subsystem $A$. This formula was derived as a central result in Ref.~\cite{bianchi17}. The most important feature is that this formula is independent of the initial state and is also independent from the metric that we use to measure the $2N_A$ dimensional volume of $M^\intercal(t)\mathcal{V}_A\subset A^*$.

When we select a subsystem, we choose a subset of $N_A$ out of $N$ pairs $(\hat a_i^\dagger,\hat a_i)$ of creation and annihilation operators. Mathematically, this corresponds to choosing a $2N_A$ dimensional subspace $A\subset V$ of the classical phase space $V$ that induces a tensor product decomposition $\mathcal{H}_N=\mathcal{H}_A\otimes\mathcal{H}_B$. For the volume formula, we only need to select a parallelepiped $\mathcal{V}_A\subset A$ by choosing $2N_A$ basis vectors $\theta^i$ with $\mathrm{span}(\theta^1,\cdots,\theta^{2N_A})=A$. Note that these $2N_A$ vectors do not, in general, coincide with the generalized eigenvectors $\Phi^j$. However, the following theorem shows that, essentially, all generic subsystems exhibit the same asymptotics of the entanglement entropy, which coincides with the maximal volume growth of a $2N_A$-dimensional parallelepiped.
\begin{theorem}[Generic entanglement production]\label{th:general}
Given a quadratic time-independent Hamiltonian $\hat{H}=\frac{1}{2}h_{ab}\hat{\xi}^a\hat{\xi}^b$ and a Gaussian initial state, the entanglement entropy for a generic subsystem $A\subset V$ with $N_A$ degrees of freedom grows asymptotically as
\begin{align}
	S_A(t)\sim \Lambda^{2N_A}_{\max}\,t+C^{2N_A}_{\max}\,\ln{(t)}\,,
\end{align}
where $\Lambda^{2N_A}_{\max}$ and $C^{2N_A}_{\max}$ are the same as in Theorem~\ref{maxvolum} on the maximal volume growth.
\end{theorem}
\begin{proof}
Using the volume formula, we reduce the problem to studying the time evolution of a $2N_A$ dimensional parallelepiped. However, in contrast to Theorem~\ref{volume} and \ref{maxvolum}, the parallelepiped is not necessarily spanned by the generalized eigenvectors $\eig^{l\sigma}_k(\kappa)$. Still, we can always decompose the $2N_A$ vectors in terms of the $2N$ generalized eigenvectors $\Phi^j$ sorted as explained in Theorem~\ref{maxvolum}. This leads to the transformation matrix $T$ with column vectors $\vec{t}_i$:
\begin{align}
\left(\begin{array}{c}
\theta^1\\[0em]
\vdots\\[0em]
\theta^{2N_A}
\end{array}\right)=\left(\begin{array}{ccc}
\smash{\fbox{\color{black}\rule[-37pt]{0pt}{1pt}$\,\,T^1_1\,\,$}} & \cdots & \smash{\fbox{\color{black}\rule[-37pt]{0pt}{1pt}$\,T^{2N}_1$}}\\[0em]
\vdots & \ddots & \vdots\\[0em]
\myunderbrace{\mystrut{1.5ex}T^1_{2N_A}}{\vec{t}_1} & \cdots & \myunderbrace{\mystrut{1.5ex}T^{2N}_{2N_A}}{\vec{t}_{2N}}
\end{array}\right)\left(\begin{array}{c}
\Phi^1\\[0em]
\vdots\\[0em]
\Phi^{2N}
\end{array}\right)\,.\color{white}{\begin{array}{c}
	.\\[0em] \\[0em] \\[0em] \\[0em] ,
	\end{array}}
\end{align}
Provided that the first $2N_A$ columns $\vec{t}_i$ are linearly independent, we can build the invertible $2N_A$-by-$2N_A$ matrix $U=\left(\vec{t}_1,\cdots,\vec{t}_{2N_A}\right)$ and move to the new basis vectors $\tilde{\theta}^j$:
\begin{align}
	\left(\begin{array}{c}
	\tilde{\theta}^1\\[0em]
	\vdots\\[0em]
	\tilde{\theta}^{2N_A}
	\end{array}\right)=
	\left(\begin{array}{cccc}
	\mathds{1}_{2N_A} & U^{-1}\vec{t}_{2N_A+1} & \cdots & U^{-1}\vec{t}_{2N}
	\end{array}\right)\left(\begin{array}{c}
	\Phi^1\\[0em]
	\vdots\\[0em]
	\Phi^{2N}
	\end{array}\right).
\end{align}
The vectors $\tilde{\theta}^j$ span another parallelepiped in the same subspace $A\subset V$, which grows asymptotically as $\mathcal{V}_A$. Moreover, each vector $\tilde{\theta}^j$ is of the form
\begin{align}
	\tilde{\theta}^j=\Phi^j+\!\!\!\!\!\!\sum^{2N}_{i=2N_A+1}\!\!\!\!\!c_i\Phi^i\,,
\end{align}
with some coefficients $c_i$. This ensures that the volume growth is dominated by the first $2N_A$ vectors $\Phi^j$, which therefore leads to the same asymptotic behavior as the one of the parallelepiped spanned by just $(\Phi^1,\cdots,\Phi^{2N_A})$. This asymptotics was already derived in Theorem \ref{maxvolum}. Only if the subsystem $A$ is such that the first $2N_A$ columns of $T$ are not linearly independent the asymptotics will change and its analysis is more complicated. However, such subsystems correspond to a subset of measure zero in the space of all subsystems. This legitimizes our statement that the entanglement production found here is generic for almost all subsystems.
\end{proof}

\section{Beyond Gaussian states and quadratic Hamiltonians\label{numerical}}
In Sec.~\ref{theoretical}, we restricted our study to quadratic time-independent Hamiltonians and Gaussian initial states. We were able to derive the exact asymptotics of the entanglement entropy thanks to powerful analytical tools to compute the following:
\begin{itemize}
	\item \textbf{Entanglement entropy of Gaussian states:}\\
	Calculating the entanglement entropy of an arbitrary state in the Hilbert space is a challenging computational problem. (Its computational cost scales with the dimension of the Hilbert space.) Here we are interested in bosonic Hilbert spaces that are infinite dimensional. In order to use standard numerical methods, we need to truncate the Hilbert space to a finite-dimensional subsector, and compute the entanglement entropy of states projected onto that subsector. This is only a good approximation if the states of interest have little overlap with the orthogonal complement of the truncated space. Analytical methods only exist for specific subclasses of states and specific systems decompositions. However, an important subclass consists of Gaussian states. For those, we presented a wide range of analytical techniques to compute and bound the entanglement entropy for arbitrary system decompositions (up to a measure zero set). In particular, writing the entanglement entropy of a state as the volume of a region $\mathcal{V}_A$ was crucial:
	\begin{align}
		S_A\approx \ln\mathrm{Vol}(\mathcal{D}_A)\,.
	\end{align}
	\item \textbf{Time evolution of quadratic Hamiltonians:}\\
	For generic Hamiltonians, the quantum time evolution is more complicated than the classical one. While the Hamiltonian equations of motion are ordinary differential equations, quantum evolution is based on the Schr\"odinger equation, which is a partial differential equation. Knowing the classical solution of the equation of motion does not help to find the quantum mechanical time evolution, unless the Hamiltonian is quadratic. In this special case, the time evolution of arbitrary connected $n$-point functions $C_{|\psi(t)\rangle}^{a_1\cdots a_n}$, for arbitrary initial states $|\psi_0\rangle$, can be computed from the classical time evolution $M(t)^a{}_b$:
	\begin{align}
		C_{|\psi(t)\rangle}^{a_1\cdots a_n}=M(t)^{a_1}{}_{b_1}\cdots M(t)^{a_n}{}_{b_n} C_{|\psi_0\rangle}^{b_1\cdots b_n}\,.
	\end{align}
	Another important property of quadratic Hamiltonians is their interplay with Gaussian states. If the Hamiltonian is quadratic, an initial Gaussian state remains Gaussian at all times. Since our analytical techniques only allow us to compute the entanglement entropy of Gaussian states, we can evolve them only with quadratic Hamiltonians to be able to study the time evolution of the entanglement entropy. Moreover, the change of the entanglement entropy is the result of the change of the volume of $\mathcal{V}_A$ under the classical Hamiltonian flow $M(t)$:
	\begin{align}
		S_A(t)\sim \ln\mathrm{Vol}\left[M^\intercal(t)\mathcal{V}_A\right]\,.
	\end{align}
\end{itemize}
It is natural to ask in which way our conclusions change in systems that violate one or both these simplifying conditions. We can distinguish the following three scenarios:
\begin{enumerate}
	\item Quadratic Hamiltonian, non-Gaussian initial states
	\item Non-quadratic Hamiltonian, Gaussian initial states
	\item Non-quadratic Hamiltonian, non-Gaussian initial states
\end{enumerate}
For the first scenario, there is already an analytic upper bound for the entanglement entropy production \cite{bianchi17}, which we test numerically and find to be saturated at long times. For scenarios two and three, we discuss in which regimes the analytically obtained behaviors for Gaussian initial states evolving under quadratic Hamiltonians still apply.

In order to study the entanglement entropy, we use a class of toy models that show different asymptotic features and which are simple enough to allow us to evaluate the entanglement entropy numerically with high accuracy. We consider two degrees of freedom with creation (annihilation) operators $\hat{a}_i^\dagger$ ($\hat{a}_i$) for $i\in \{1,2\}$. We truncate the corresponding Fock space $\mathcal{H}_V$ of two degrees of freedom to the following $(T+1)$-dimensional subsector
\begin{align}\label{span}
	\mathcal{H}_\mathrm{trunc}=\mathrm{span}\{|n,n\rangle\,\text{with}\,0\leq n\leq T\}\subset\mathcal{H}_V
\end{align}
where we choose different truncation sizes with up to $T=100,\!000$. This truncation can be used only if the time evolution results in a state that mostly remains in the truncated subspace. We therefore require that the Hamiltonian $\hat{H}$ commutes with the difference number operator $\hat{n}_1-\hat{n}_2=\hat{a}_1^\dagger \hat{a}_1-\hat{a}_2^\dagger \hat{a}_2$. This ensures that the time evolution preserves the subspace whose states $|\psi\rangle$ satisfy $(\hat{n}_1-\hat{n}_2)|\psi\rangle=0$. Specifically, we choose
\begin{align}\label{twosite_hamiltonian}
	\hat{H}(\Delta,U)=(\hat{n}_1+\hat{n}_2)+\Delta\left(\hat{a}_1^\dagger \hat{a}_2^\dagger+\hat{a}_1\hat{a}_2\right)+ \frac{U}{2}\left(\hat{n}_1^2+\hat{n}_2^2\right)
\end{align}
as our Hamiltonian, with free parameters $\Delta$ and $U$. 

The noninteracting part of this Hamiltonian can be expressed in the basis of  $(\hat{q}_1,\hat{q}_2,\hat{p}_1,\hat{p}_2)$, which was introduced in Eq.(\ref{new basis}), as
\begin{align}\label{twosite_hamiltonian_qpbasis}
	\begin{gathered}
		\hat{H}(\Delta,0)=\frac{1}{2}(\hat{p}_{1}^{2}+\hat{p}_{2}^{2}-2\Delta \hat{p}_{1}\hat{p}_{2})+\frac{1}{2}(\hat{q}_{1}^{2}+\hat{q}_{2}^{2}+2\Delta \hat{q}_{1}\hat{q}_{2}), \\
		\quad\text{with}\\
		h\equiv\left(\begin{array}{cccc}
		1 & -\Delta & 0 & 0\\[0em]
		-\Delta & 1 & 0 & 0\\[0em]
		0 & 0 & 1 & \Delta\\[0em]
		0 & 0 & \Delta & 1
		\end{array}\right)\,\Rightarrow\, K\equiv\left(\begin{array}{cccc}
		0 & 0 & \Delta & -1\\[0em]
		0 & 0 & -1 & \Delta\\[0em]
		\Delta & 1 & 0 & 0\\[0em]
		1 & \Delta & 0 & 0
		\end{array}\right)\,.
	\end{gathered}
\end{align}	

Different choices of those parameters correspond to different classes of Hamiltonians.
\begin{itemize}
	\item[(1)] \textbf{Quadratic Hamiltonian} $\hat{H}(\Delta,0)$:\\
	The eigenvalues of the matrix $K$ are: $K=\{\sqrt{\Delta ^{2}-1},\sqrt{\Delta ^{2}-1},-\sqrt{\Delta ^{2}-1},-\sqrt{\Delta ^{2}-1}\}$. 
	Hence,  the Hamiltonian is : (a) unstable for $|\Delta|>1$ (all eigenvalues of $K$ are real), (b) stable for $|\Delta|<1$ (all eigenvalues of $K$ are imaginary), and (c) metastable for $|\Delta|=1$ (all eigenvalues of $K$ are zero). We then consider: (a) $\Delta=1.5$, (b) $\Delta=0.5$, and (c) $\Delta=1$. The entanglement entropy is computed by evaluating Eq.~(\ref{eq:master}), with $J(t)=M(t)J_0 M^{-1}(t)$ and $M(t)=e^{tK}$. The initial complex structure $J_0$ can be computed from the initial $2$-point function $(G_0)^{ab}$ via $(J_0)^a{}_{b}=-(G_0)^{ac}\omega_{cb}$.
	\item[(2)] \textbf{Non-quadratic Hamiltonian} $\hat{H}(\Delta,U)$:\\
	We investigate the effect of non-quadratic perturbations by adding the quartic term $(\hat{n}^2_1+\hat{n}^2_2)$ with a small prefactor $U>0$. This Hamiltonian is always bounded from below. $\hat{H}(\Delta,0)$ is the expansion of $\hat{H}(\Delta,U)$ to quadratic order in creation and annihilation operators.
\end{itemize}
We compute the time evolution of the entanglement entropy for various initial states:
\begin{itemize}
	\item[(A)] \textbf{Fock states:} $|n,n\rangle$.
	\item[(B)] \textbf{Gaussian states:} $\varrho_{|n,n\rangle}$.\\
	We construct the Gaussian part of $|n,n\rangle$, namely, $\varrho_{|n,n\rangle}$, numerically. The state $\varrho_{|n,n\rangle}$ has the same $1$- and $2$-point functions as $|n,n\rangle$, but all higher $n$-point functions are constructed via Wick's theorem. Note that this may mean that $\varrho_{|n,n\rangle}$ is not a pure state. The entanglement entropy $S_A(\varrho_{|n,n\rangle})$ is computed using $\varrho_{|n,n\rangle}$ to evaluate our analytic expressions, and is compared to the entanglement entropy $S_A(|n,n\rangle)$. They satisfy the inequality $S_A(|n,n\rangle)\leq S_A(\varrho_{|n,n\rangle})$ at all times \cite{bianchi17}.
	\item[(C)] \textbf{Random state:} $|\mathrm{ran}\rangle$.\\
	In order to study the behavior of the entanglement entropy for random states in the truncated Hilbert space, we generate a random state
	\begin{align}
		|\mathrm{ran}\rangle=\frac{1}{\sqrt{\sum^{100}_{n=0}c_n^{\,2}}}\sum^{100}_{n=0}c_n\,|n,n\rangle\,,
	\end{align}
	where $c_n$ are selected randomly between 0 and 1 with uniform probability.
\end{itemize}
We compare our analytical results for Gaussian states and quadratic Hamiltonians (lines) with numerical computations of various Gaussian and non-Gaussian states evolving under quadratic and non-quadratic Hamiltonians (symbols). Figure~\ref{fig1}(a) illustrates that our numerical results agree perfectly with the analytical ones for the Gaussian initial state $|0,0\rangle$ evolving under a quadratic Hamiltonian. For the non-Gaussian initial state $|1,1\rangle$, see Fig.~\ref{fig1}(b), we compare the numerical results to the analytical ones for $\varrho_{|1,1\rangle}$. The latter serve as an upper bound to the former ones \cite{bianchi17}.

\begin{figure}[t]
	\begin{center}
		\begin{tikzpicture}
		\draw (0,0) node{\includegraphics[width=\linewidth]{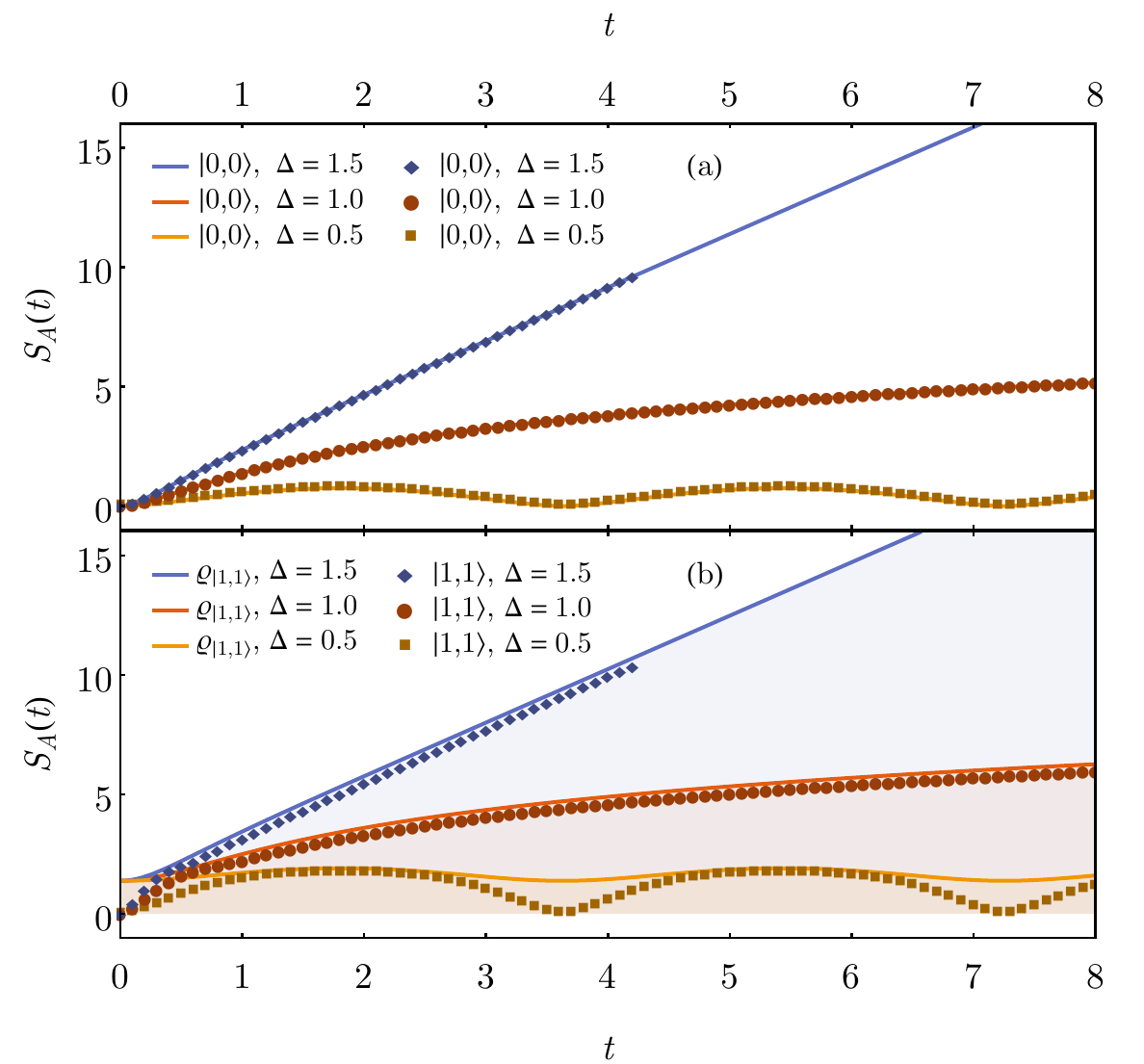}};
		\end{tikzpicture}
	\end{center}
	\vspace{-.8cm}
	\caption{Quadratic Hamiltonians with Gaussian and non-Gaussian initial states. We compare numerical (symbols) and analytical (lines) results for the unstable ($\Delta=1.5$), stable ($\Delta=0.5$), and metastable ($\Delta=1.0$) Hamiltonians. (a) Gaussian initial state $|0,0\rangle$. (b) Non-Gaussian state initial state $|1,1\rangle$ compared to analytical results for the Gaussian state $\varrho_{|1,1\rangle}$. The Hilbert space truncation is $T=50,\!000$.}
\label{fig1}
\end{figure}

\subsection{Non-Gaussian initial states evolving under quadratic Hamiltonians}
\begin{figure}[!t]
	\begin{center}
		\begin{tikzpicture}
		\draw (0,0) node{\includegraphics[width=\linewidth]{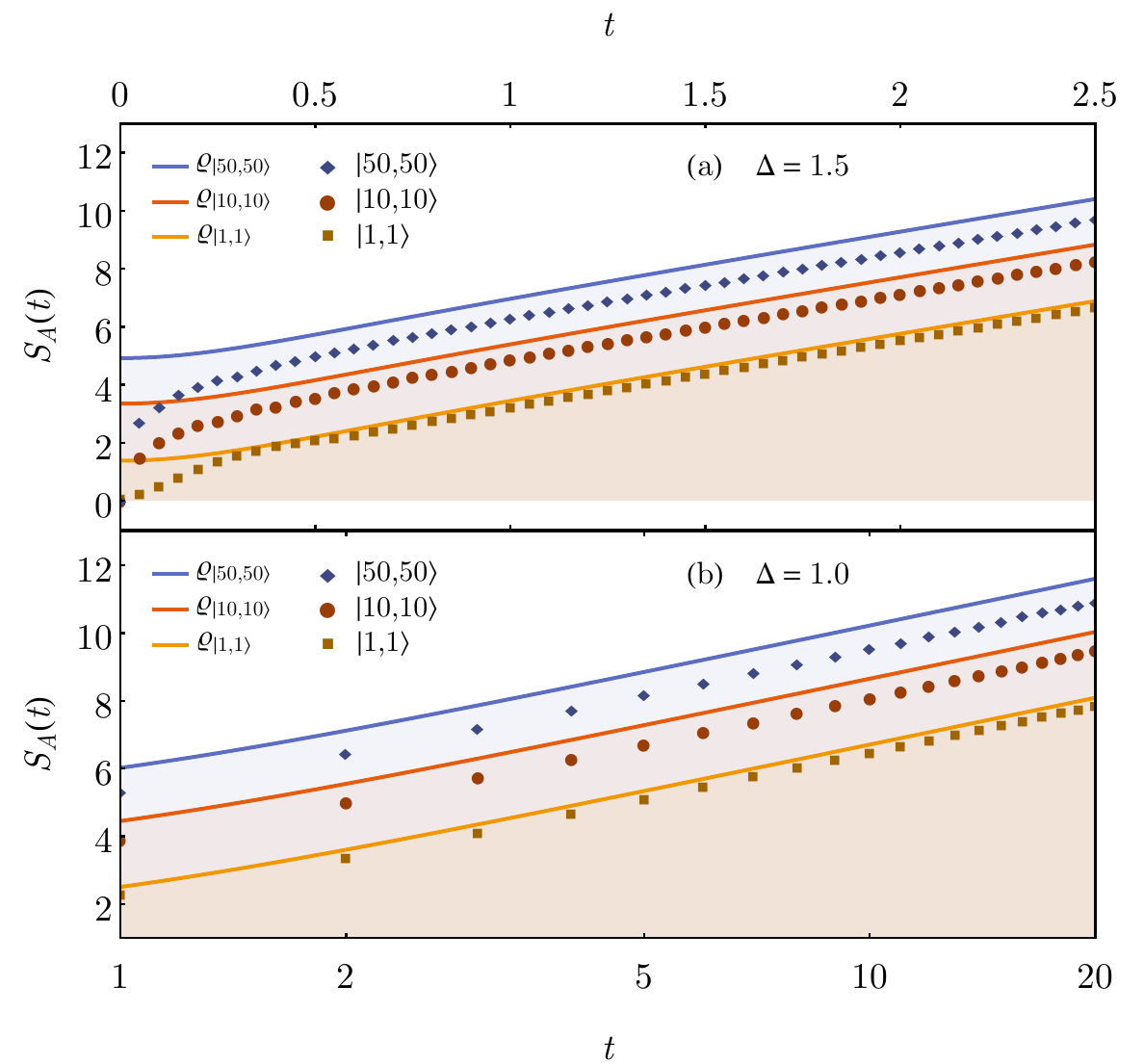}};
		\end{tikzpicture}
	\end{center}
	\vspace{-.8cm}
	\caption{Entanglement production for non-Gaussian initial states. We compare analytical (lines) with numerical (symbols) results for: (a) an unstable Hamiltonian ($\Delta=1.5$) for a truncation $T=50,\!000$, and (b) a metastable Hamiltonian ($\Delta=1.0$) for a trunction $T=100,\!000$.}
\label{fig2}
\end{figure}

\begin{figure}[!b]
	\begin{center}
		\begin{tikzpicture}
		\draw (0,0) node{\includegraphics[width=\linewidth]{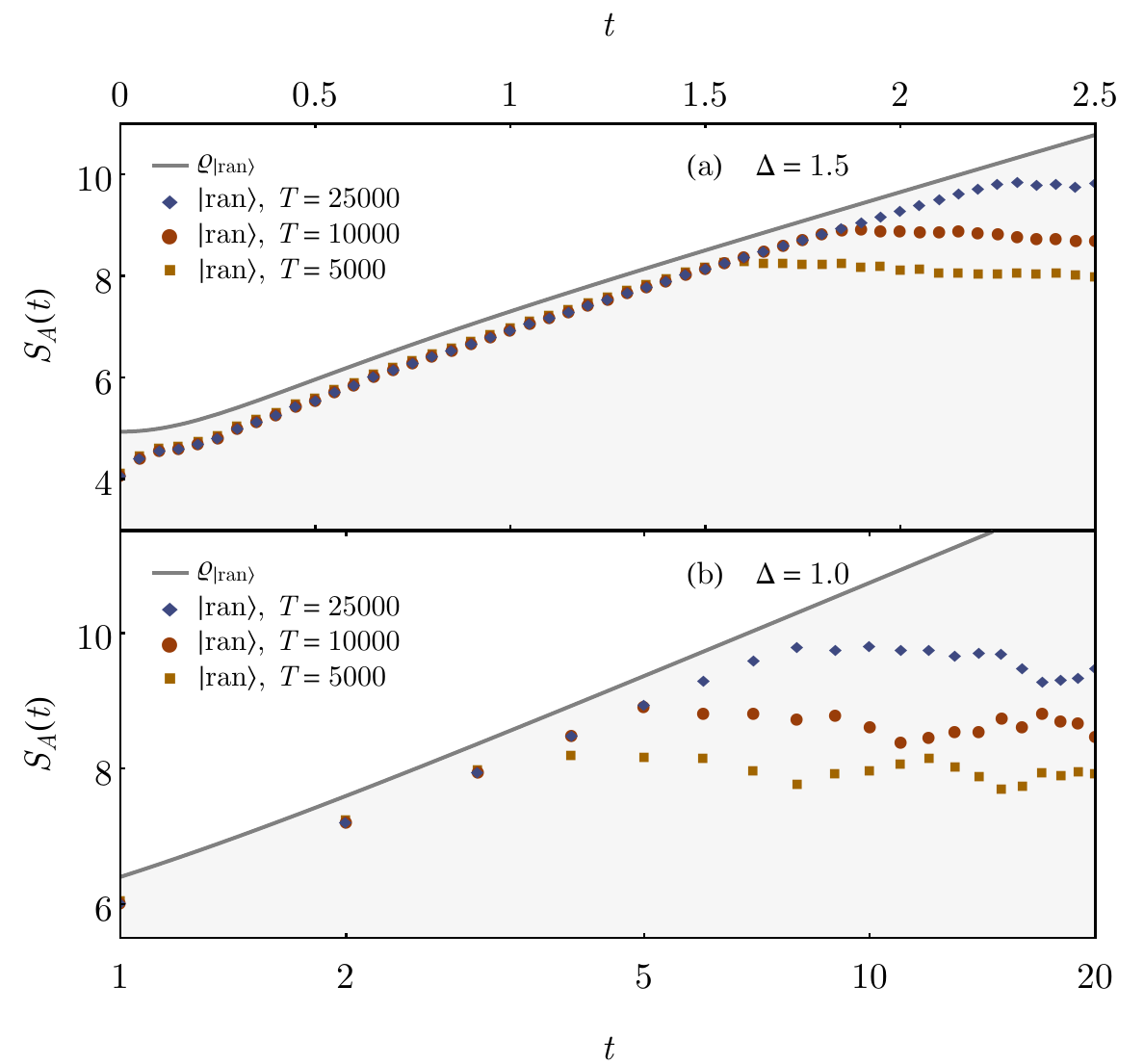}};
		\end{tikzpicture}
	\end{center}
	\vspace{-.8cm}
	\caption{Entanglement production for a random non-Gaussian initial state. We compare analytical (lines) with numerical (symbols) results for: (a) an unstable ($\Delta=1.5$), and (b) a metastable ($\Delta=1.0$) Hamiltonian when the Hilbert space truncation $T$ is varied.}
\label{fig3}
\end{figure}
In Fig.~\ref{fig2}, we present a more comprehensive study of the entropy production during the dynamics of non-Gaussian initial states under quadratic Hamiltonians. We compare the entanglement entropy $S_A(t)$ for the non-Gaussian initial states $|1,1\rangle$, $|10,10\rangle$, and $|50,50\rangle$ with the one of the corresponding Gaussian initial states $\varrho_{|n,n\rangle}$, for an unstable [Fig.~\ref{fig2}(a)] and a metastable [Fig.~\ref{fig2}(b)] Hamiltonian. The numerical results show that the corresponding Gaussian initial states provide an upper bound for the non-Gaussian ones. More importantly, we find that $S_A(t)$ for non-Gaussian states shows the characteristic asymptotic behavior of the entropy derived analytically for Gaussian states. Namely, a linear increase for unstable Hamiltonians and a logarithmic one for metastable ones. This leads us to conjecture that for both, unstable and metastable Hamiltonians, the asymptotic predictions for Gaussian initial states apply to any initial state. 

Results for the entropy production during the dynamics of an initial random state, for different truncations of the Hilbert space, are presented in Fig.~\ref{fig3}. The asymptotic behavior (with increasing $T$) of the entanglement entropy in Fig.~\ref{fig3} is identical to the one in Fig.~\ref{fig2}. Finite values of $T$ lead to a saturation of the entanglement entropy at times that grow with increasing $T$.

\subsection{Gaussian initial states evolving under non-quadratic Hamiltonians}
\begin{figure}[!b]
	\begin{center}
		\begin{tikzpicture}
		\draw (0,0) node{\includegraphics[width=\linewidth]{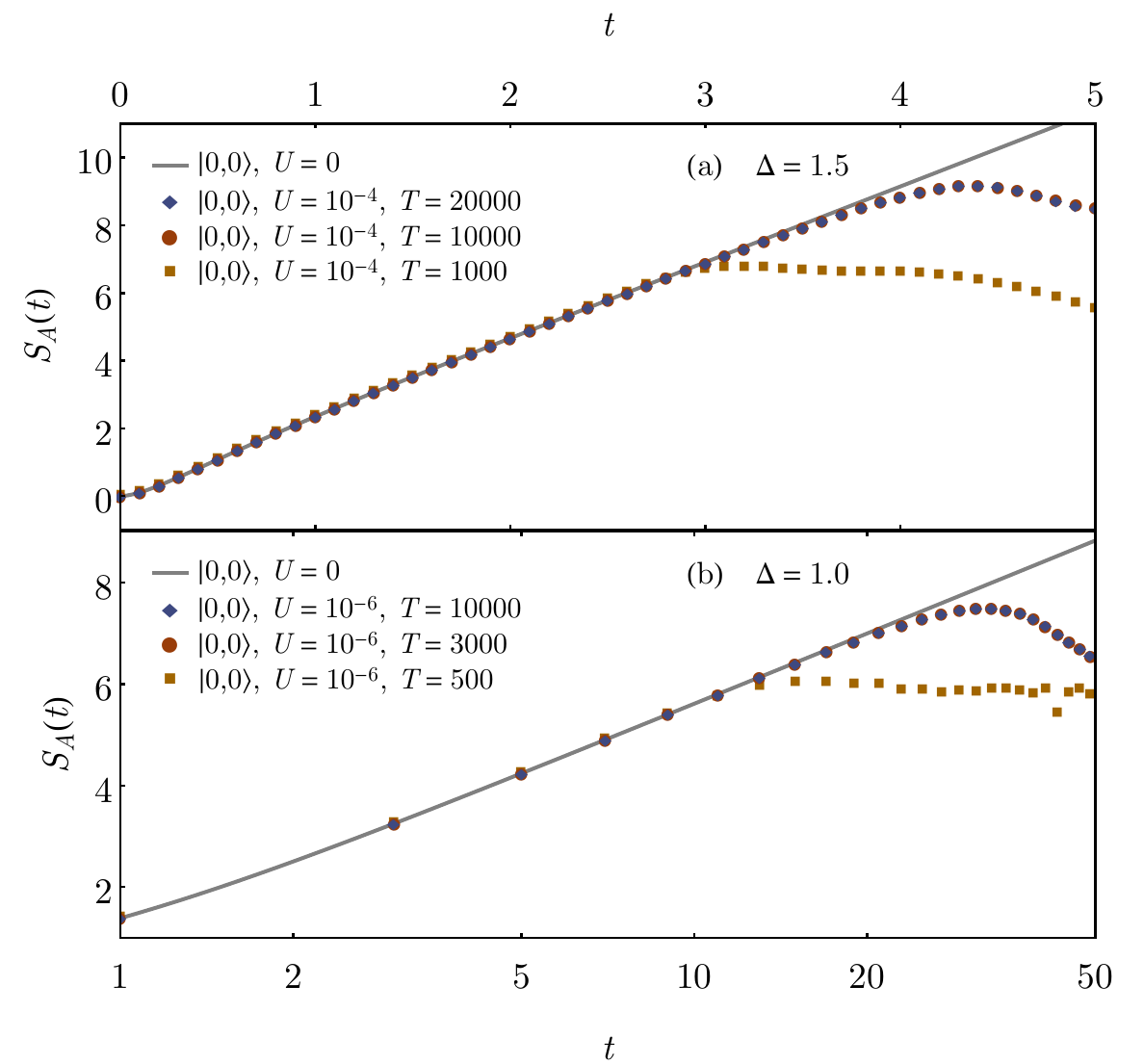}};
		\end{tikzpicture}
	\end{center}
	\vspace{-.8cm}
	\caption{Entanglement production for a Gaussian initial state ($|0,0\rangle$) evolving under non-quadratic Hamiltonians. The entanglement production under non-quadratic Hamiltonians (symbols) is compared to the one under their quadratic parts (lines). We also vary the Hilbert space truncation $T$. The results for the two largest values of $T$ overlap in each panel.}
\label{fig4}
\end{figure}

\begin{figure}[!t]
	\begin{center}
		\begin{tikzpicture}
		\draw (0,0) node{\includegraphics[width=\linewidth]{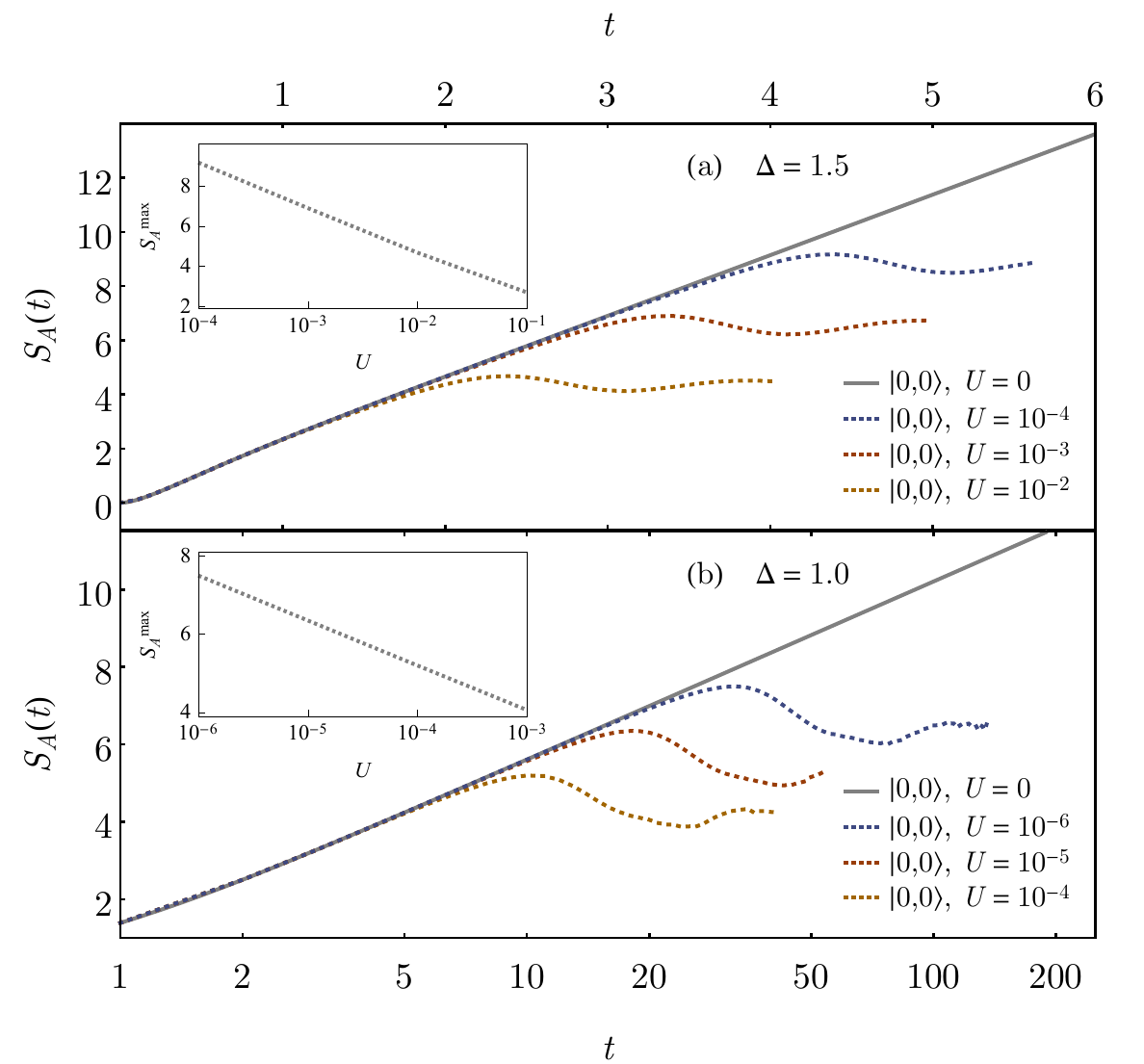}};
		\end{tikzpicture}
	\end{center}
	\vspace{-.8cm}
	\caption{Entanglement entropy saturation for a Gaussian initial state ($|0,0\rangle$) evolving under non-quadratic Hamiltonians. In each panel, we compare results obtained for different values of $U$. In the insets, we plot the entanglement entropy at the first maximum during the time evolution, $S_A^{\mathrm{max}}$, as a function of $U$. The Hilbert space truncation is $T=50,\!000$.}
\label{fig5}
\end{figure}

Next, we consider Gaussian initial states evolving under non-quadratic Hamiltonians. Because of the non-quadratic part, the overall Hamiltonian is bounded from below, which implies that the entanglement entropy is bounded from above by the thermal entropy at the energy of the time-evolving state. We are interested in the regime in time in which the entanglement production exhibits the same behavior as for the nearby quadratic Hamiltonian. It is expected that, as long as the non-quadratic part can be neglected compared to the quadratic one for a given initial state, the entanglement production will be dominated by the quadratic part. However, there will always be a time at which the entanglement entropy will depart from the result for the quadratic part. As mentioned before, an initial Gaussian state does not remain Gaussian under the time evolution with a non-quadratic Hamiltonian.

In Fig.~\ref{fig4}, we show the entanglement production for an initial Gaussian state ($|0,0\rangle$) evolving under unstable [Fig.~\ref{fig4}(a)] and metastable [Fig.~\ref{fig4}(b)] Hamiltonians. We compare the numerical results with the analytical prediction for the entanglement production of the corresponding quadratic Hamiltonian. We find a perfect agreement at short and intermediate times, but eventually the entanglement entropy in the non-quadratic systems saturates. In Fig.~\ref{fig4}, we report results for various Hilbert space truncations to demonstrate that, for sufficiently large truncations, the saturation value is independent of the truncation chosen. Figure~\ref{fig5} shows how the saturation value depends on $U$ for sufficiently large Hilbert space truncations. In particular, in the insets we plot the entanglement entropy at the first maximum during the time evolution, $S_A^{\mathrm{max}}$, as a function of $U$. In the regime studied, $S_A^{\mathrm{max}}$ decreases near logarithmically with $U$.

\subsection{Non-Gaussian initial states evolving under non-quadratic Hamiltonians}
\begin{figure}[!t]
	\begin{center}
		\begin{tikzpicture}
		\draw (-4,0) node{\includegraphics[width=\linewidth]{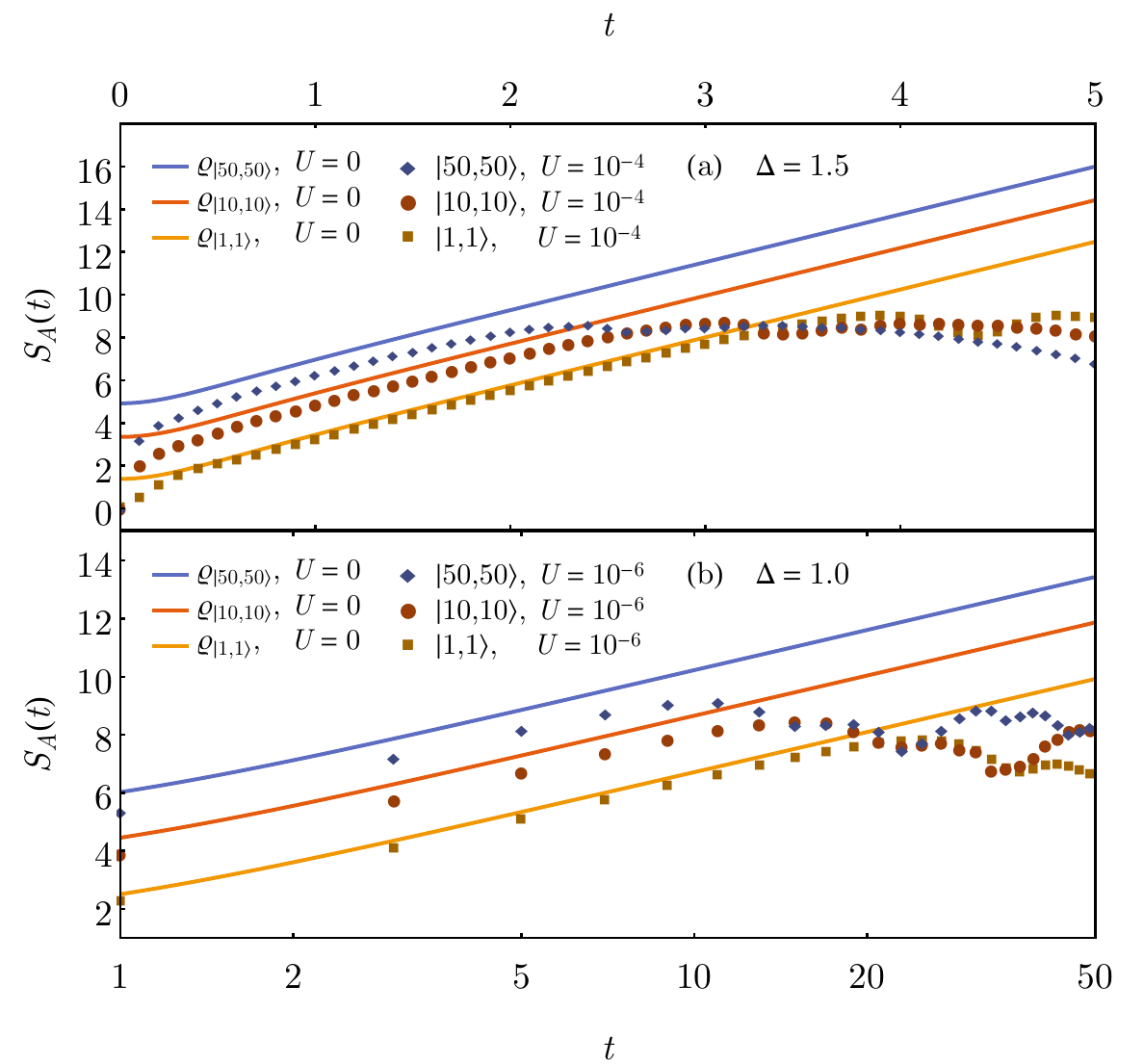}};
		\end{tikzpicture}
	\end{center}
	\vspace{-.8cm}
	\caption{Entanglement production for non-Gaussian initial states evolving under non-quadratic Hamiltonians. We compare analytical results for dynamics under quadratic Hamiltonians (lines) with numerical results (symbols) for: (a) an unstable ($\Delta=1.5$), and (b) a metastable ($\Delta=1.0$) Hamiltonian. The Hilbert space truncation is $T=50,\!000$.}
\label{fig6}
\end{figure}

\begin{figure}[!b]
	\begin{center}
		\begin{tikzpicture}
		\draw (4,0) node{\includegraphics[width=\linewidth]{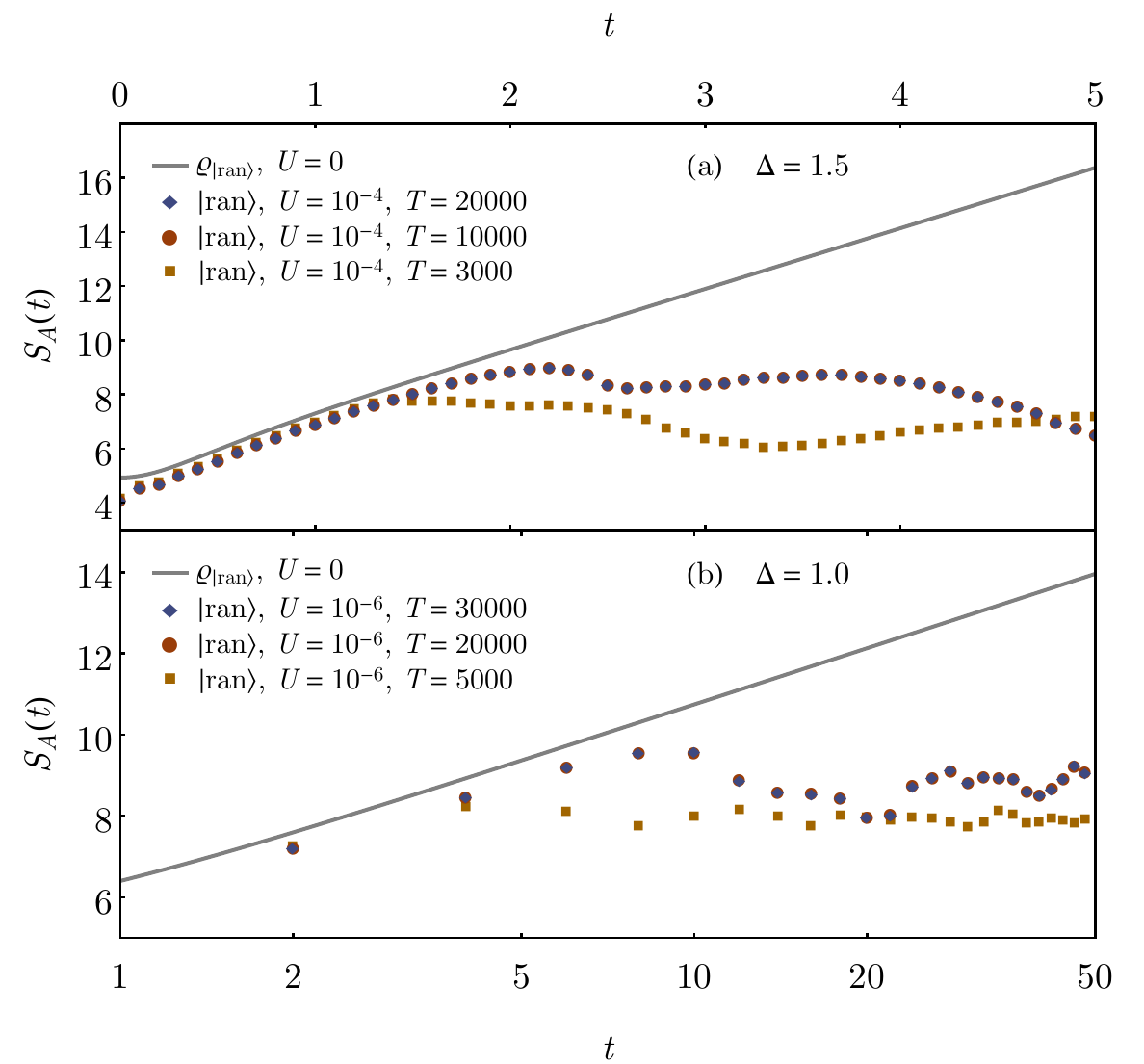}};
		\end{tikzpicture}
	\end{center}
	\vspace{-.8cm}
	\caption{Entanglement production for a random non-Gaussian initial state evolving under non-quadratic Hamiltonians. We compare the analytical results for dynamics under quadratic Hamiltonians (lines) with the numerical results (symbols) for: (a) an unstable ($\Delta=1.5$), and (b) a metastable ($\Delta=1.0$) Hamiltonian when the Hilbert space truncation $T$ is varied. The results for the two largest values of $T$ overlap in each panel.}
\label{fig7}
\end{figure}

To close the numerical part of our study, we explore the evolution of non-Gaussian initial states under non-quadratic Hamiltonians. Clearly, if the initial state has an entanglement entropy that is close to that of the thermal state with the same energy, there is no intermediate regime in the entanglement entropy dynamics exhibiting a linear or logarithmic growth (the entanglement entropy cannot change much). Hence, we focus on non-Gaussian initial states whose entanglement entropy is much smaller than the thermal one.

The results for the entanglement production at short and intermediate times in such states are very similar for the dynamics under non-quadratic Hamiltonians (see Figs.~\ref{fig6} and~\ref{fig7}) and quadratic ones (see Figs.~\ref{fig2} and~\ref{fig3}). The main difference between them is the $U$ dependent saturation that occurs in the former. At intermediate times they are all very similar, and are bounded from above by the analytical predictions for Gaussian initial states.

\section{Discussion\label{discussion}}

\subsection{Relation to global quenches}
Entanglement entropy dynamics after uantum quenches have been extensively studied in free field theories \cite{Calabrese:2005in,calabrese2007quantum,cotler2016entanglement} and in many-body quantum systems \cite{dechiara2006entanglement,fagotti2008evolution,eisler2008entanglement,lauchli2008spreading,kim2013ballistic,alba2017entanglement}. In the usual set up for global quenches, an initial state $|\psi_0\rangle$ (commonly the ground state of some initial Hamiltonian) is evolved with a time-independent Hamiltonian (of which $|\psi_0\rangle$ is not an eigenstate). In the context of free field theories, $|\psi_0\rangle$ is usually a translationally invariant state with short-ranged entanglement. The time evolution of an arbitrary initial state studied in this work can be considered to be the result of a global quench. We have focused on initial states with low entanglement.

Evolving a low-entanglement state with a local Hamiltonian generally leads to a propagation of correlations that results in a linear growth of the entanglement entropy (an important exception being many-body localized systems~\cite{nandkishore_huse_review_15}). As mentioned before, in systems with free quasi-particles, the linear growth can be understood to be the result of the entanglement produced by the free propagation of the quasi-particles, with the entanglement production rate determined by the propagation speed \cite{Calabrese:2005in}. 

At this point, it is important to emphasize that the physical mechanism of linear entanglement production due to {\it instabilities}, presented in this paper, is manifestly different from the linear entanglement production due to {\it propagation of quasi-particles}. The distinction is best explained for bosonic Gaussian states, for which the entanglement entropy can be explicitly decomposed into a sum in which each addend is associated with a specific entangled pair consisting of a single degree of freedom in the subsystem and one in the complement. This leads to a pair of eigenvalues $\pm\ii\nu_i$ of the restricted complex structure $[J]_A$ with a contribution to the entanglement entropy given by
\begin{equation}
	S(\nu_i)=\frac{\nu+1}{2}\ln\frac{\nu_i+1}{2}-\frac{\nu_i-1}{2}\ln\frac{\nu_i-1}{2}\,.
\end{equation}
We can now distinguish the following mechanisms:
\begin{itemize}
	\item \textbf{Entropy production due to instabilities}\\
	This is the mechanism relevant to the systems studied in this paper. The entanglement entropy grows because the time-evolution leads to an exponential growth $\nu_i\sim \ee^{\lambda_i t}$ of certain eigenvalues of $[J]_A$ corresponding to the unstable direction of the time-evolution. Thus, each addend of $S_A(t)=\sum_i S(\nu_i)$ grows as $S(\nu_i)\sim \lambda_i t$ leading to a production rate $\Lambda_A=\sum_i\lambda_i$. This mechanism was extensively studied in Ref.~\cite{bianchi17}, where the connection between instabilities and Lyapunov exponents of classical dynamical systems was established. We emphasize that this mechanism only works for bosons as the entanglement entropy of bosonic pairs can be arbitrarily large. The entanglement entropy of fermionic pairs is bounded from above by $\ln2$.
	\item \textbf{Production due to propagation}\\
	This is the mechanism studied in Refs.~\cite{Calabrese:2005in,calabrese2007quantum,cotler2016entanglement} for free bosonic field theories. Here, the entanglement entropy grows because successively more degrees of freedom become entangled with each other, but for each entangled pair the entanglement entropy is bounded. For Gaussian states, this means that most eigenvalues $\nu_i$ start off close to $1$ with $S(1)=0$. The time-evolution results in individual eigenvalues $\nu_i$ successively evolving from $1$ to a maximal value $\nu_{\max}>1$. The speed of propagation $c$ determines how many eigenvalues move from $1$ to $\nu_{\max}$ in a given time interval. In this case, the entanglement entropy at time $t$ is given by $S_A(t)=\sum_iS(\nu_i)=c\,t\,S(\nu_{\max})$, so that the production rate is $c\,S(\nu_{\max})$. Clearly, this mechanism requires a large number of degrees of freedom that become successively entangled and, for any finite system, the entanglement entropy eventually saturates at $S_{\mathrm{sat}}=\mathrm{min}(N_A,N_B)\,S(\nu_{\max})$ when all possible degrees of freedom in $A$ and $B$ are maximally entangled with each other.
\end{itemize}
The two mechanisms require distinct settings to be tested. To study entanglement production due to instabilities, the Hilbert space per degree of freedom needs to be very large to capture the exponential growth, but it is less important to have systems with many degrees of freedom. In fact, computational limitations forced us in the present work to choose the minimal system with just two degrees of freedom. In contrast, entropy production due to entanglement propagation can be studied with relatively small Hilbert spaces per degree of freedom, but it is essential to have a large number of degrees of freedom to actually observe the effect of the propagation.

\subsection{Application to periodically driven systems}
We introduced the notion of a Floquet Hamiltonian $\hat{H}_{\mathrm{F}}$, which is derived from a Hamiltonian with periodic time-dependence, $\hat{H}(t+T)=\hat{H}(t)$. This allows one to study the entanglement entropy $S_A(t_n)$ at times $t_n=nT$. Even though one does not know a priori how much the entanglement entropy fluctuates during the intermediate times, the stroboscopic analysis typically suffices to predict the correct asymptotics and the production rate. In this case, the eigenvalues of the symplectic generator $K$ are referred to as Floquet exponents. While their imaginary part encodes the contributing phase shifts that are added each period, their real parts are the Lyapunov exponents $\lambda_i$ that determine directly the asymptotic behavior of the entanglement entropy.

\subsection{Conjecture for non-Gaussian initial states evolving under quadratic Hamiltonians}
Based on our numerical results for the entanglement entropy of quadratic Hamiltonians with non-Gaussian initial states, we put forward the following conjecture.
\begin{conjecture}\label{con1}
Given a bosonic system with a time-independent quadratic Hamiltonian, the time evolution of the entanglement entropy $S_A(t)$ for \underline{arbitrary} initial states with finite average energy\footnote{The requirement of finite average energy is with respect to positive-definite quadratic Hamiltinians and excludes unphysical states, such as $|\psi\rangle=\frac{\sqrt{6}}{\pi}\sum^\infty_{n=1}\frac{1}{n}|n,n\rangle$ with $\langle\psi|\hat{n}_i|\psi\rangle=\infty$. It is equivalent to requiring a finite covariance matrix of the state.} has the same asymptotic behavior found for Gaussian initial states, namely
\begin{align}
	S_A(t)=\Lambda_A\,t+C_A\,\ln(t)+X_A(t)\,,
\end{align}
where only $X_A(t)$ depends on the initial state.
\end{conjecture}

We emphasize that our numerical study should be seen as a first step towards a better understanding of how broadly this conjecture applies. Due to computational limitations, our numerical calculations were limited to systems with just two degrees of freedom, but with a large dimension of the truncated Hilbert space. We expect qualitatively similar results for systems with more degrees of freedom, for which we presented the general production Theorem \ref{th:general}. An analytical proof of our conjecture will require novel tools to bound the entanglement entropy of non-Gaussian states from below, possibly in the spirit of power inequalities as presented in Ref.~\cite{de2015multimode}.

\begin{acknowledgments}
LH thanks Giacomo De Palma, Albert Werner and Jens Eisert for discussions. This work was supported by NSF Grants No.~PHY-1404204 (EB) and PHY-1707482 (MR), a Frymoyer fellowship and a Mebus fellowship (LH), by the Army Research Office Grant No.~W911NF1410540 (RM and MR). The computations were carried out at the Institute for CyberScience at Penn State. This research was supported in part by the Perimeter Institute for Theoretical Physics.
\end{acknowledgments}

\appendix
\section{Jordan normal form of real matrices\label{jordan}}
Most square matrices are diagonalizable.\footnote{By ``most'', we mean all but a subset of measure zero. In particular, all Hermitian and anti-Hermitian matrices are diagonalizable. This because those matrices are special cases of matrices that commute with their adjoint (with respect to the chosen inner product). Those matrices are called normal and it can be proved that all normal matrices are diagonalizable.} However, there exists a subset of matrices that cannot be brought into diagonal form by an equivalence transformation and so they are non-diagonalizable. We consider the most general decomposition that can be applied to all real square matrices, even if they are not diagonalizable.

The Jordan normal form is standard material in linear algebra \cite{weintraub2009jordan}. However, it is useful to review the construction procedure using a convention that illuminates its application in this paper.
\begin{theorem}[Jordan normal form]
Given a real linear map\footnote{We refer to this map as the transpose of $K$ to match the notation in the main text where $K^\intercal$ plays an important role. Due to the fact that the eigenvalues of $K$ and $K^\intercal$ are the same, the full analysis can be equivalently carried out for $K$.} $K^\intercal: V^*\to V^*$ on a finite dimensional real vector space $V^*$ with $n$ distinct eigenvalues $\kappa_1,\cdots,\kappa_n$ satisfying $\mathrm{Im}(\kappa_i)\geq 0$ (meaning, for complex eigenvalues we only take the one with positive imaginary part), there always exists a basis, such that the matrix representation of $K$ is block-diagonal as
\begin{align}
\begin{split}
	&K^\intercal\equiv\left(\begin{array}{cccc}
	\boxed{J(\kappa_1)} & & &\\[0em]
	& \boxed{J(\kappa_2)} & & \\[0em]
	& &\ddots &\\[0em]
	& & & \boxed{J(\kappa_n)}
	\end{array}\right),\quad\text{with}\\
	&J(\kappa)\equiv\left(\begin{array}{ccc}
	\boxed{J_1(\kappa)} & &\\[0em]
	 &\ddots &\\[0em]
	 & & \boxed{J_{j_\kappa}(\kappa)}
	\end{array}\right)
\end{split}
\end{align}
where the Jordan blocks $J_k(\kappa)$ are real square matrices of the following form:
\begin{itemize}
	\item \textbf{Real eigenvalue $\kappa=\lambda\in\mathbb{R}$}\\
	For real eigenvalues $\kappa$, the Jordan block $J_k(\kappa)$ can have an arbitrary dimension $j_i(\kappa)$ and takes the following form
	\begin{align}
		\qquad J_k(\kappa)=\underbrace{\left(\begin{array}{cccc}
		\lambda & & &\\[0em]
		& \lambda & & \\[0em]
		& &\ddots &\\[0em]
		& & & \lambda
		\end{array}\right)}_{A_k(\kappa)}+\underbrace{\left(\begin{array}{cccc}
		0 & 1& &\\[0em]
		& 0 & \ddots& \\[0em]
		& &\ddots &1\\[0em]
		& & & 0
		\end{array}\right)}_{C_k(\kappa)}\,.
	\end{align}
	\item \textbf{Complex eigenvalue $\kappa=\lambda+\ii\omega\in\mathbb{C}$}\\
	For a complex eigenvalue $\kappa$, the Jordan block $J_k(\kappa)$ must have an even dimension $\dim J_k(\kappa)$ and takes the following form
	\begin{align}
	\small
	J_k(\kappa)=&\underbrace{\left(\begin{array}{cccc}
		\lambda \mathds{1}_2 & & &\\[0em]
		& \lambda \mathds{1}_2 & & \\[0em]
		& &\ddots &\\[0em]
		& & & \lambda \mathds{1}_2
		\end{array}\right)}_{A_k(\kappa)}\\
	&+\underbrace{\left(\begin{array}{cccc}
		\omega \mathds{B}_2 & & &\\[0em]
		& \omega \mathds{B}_2 & & \\[0em]
		& &\ddots &\\[0em]
		& & & \omega \mathds{B}_2
		\end{array}\right)}_{B_k(\kappa)}
	+\underbrace{\left(\begin{array}{cccc}
		0 & \mathds{1}_2& &\\[0em]
		& 0 & \ddots& \\[0em]
		& &\ddots &\mathds{1}_2\\[0em]
		& & & 0
		\end{array}\right)}_{C_k(\kappa)}\,,\nonumber
	\end{align}
	where each entry represents a 2-by-2 matrix block that is either proportional to the identity $\mathds{1}_2$ or to the antisymmetric matrix
	\begin{align}
		\mathds{B}_2=\left(\begin{array}{cc}
		0 & 1\\[0em]
		-1 & 0
		\end{array}\right)\,.
	\end{align}
\end{itemize}
\end{theorem}
\begin{proof}
We construct explicitly a basis in which $K$ takes the Jordan normal form as described above. The construction only involves computing eigenvalues and solving linear equations to find the kernel of a matrix.\\
\textbf{Step 1.} Compute the $n$ distinct eigenvalues $\kappa_1,\cdots,\kappa_n$ with $\mathrm{Im}(\kappa_i)\geq 0$ by finding the roots of the characteristic polynomial
\begin{align}
	\chi(\kappa)=\det(K^\intercal-\kappa\mathds{1})\,.
\end{align}
Complex eigenvalues always appear in conjugate pairs and we only include the one with positive imaginary part in our list.\\
\textbf{Step 2.} For every eigenvalue $\kappa$, we construct the corresponding block $J(\kappa)$ and its Jordan blocks $J_k(\kappa)$. For this, we need to study generalized eigenspaces. The generalized eigenspace of order $m$ is defined as
\begin{align}
E^{(m)}(\kappa)=\ker(K^\intercal-\kappa\mathds{1})^m\,.
\end{align}
The first order eigenspace $E^{(1)}(\kappa)$ is just the regular eigenspace, but higher order eigenspaces are larger. The number $j_\kappa$ of distinct Jordan blocks $J_k(\kappa)$ is given by $j_\kappa=\dim E^{(1)}(\kappa)$. The number of Jordan blocks of dimension $m$ is given by
\begin{align}
t_m=2\dim E^{(m)}(\kappa)-\dim E^{(m-1)}(\kappa)-\dim E^{(m+1)}(\kappa),
\end{align}
and each Jordan block is generated by a highest weight vector $\eig_k(\kappa)$. A highest weight vector generates a sequence of $m_k(\kappa)$ vectors
\begin{align}
\eig^1_k(\kappa),\cdots,\eig^{m_k(\kappa)}_k(\kappa)\,.
\end{align}
by repeatedly applying $(K^\intercal-\kappa\mathds{1})$ to it: $\eig^{l}_k(\kappa)=(K^\intercal-\kappa\mathds{1})^{m_k(\kappa)-l}\,\eig^{l}_k(\kappa)$. One can apply the following induction to find all highest weight vectors:
\begin{enumerate}
	\item We start with the largest $m$, for which $t_m\neq 0$ and select $t_m$ linearly independent vectors in $\eig_k(\kappa)\in E^{(m)}(\kappa)\setminus E^{(m-1)}(\kappa)$. After this, we construct for each vector $\eig_k(\kappa)$ the corresponding sequence
	\begin{align}
	\eig^1_k(\kappa),\cdots,\eig^{m_k(\kappa)}_k(\kappa)\,,
	\end{align}
	by repeatedly applying $(K^\intercal-\kappa\mathds{1})$ to it. 
	\item Next, we can go to $m-1$ and select $t_{m-1}$ additional vectors $\eig_k(\kappa)$ out of $E^{(m-1)}(\kappa)\setminus E^{(m-2)}(\kappa)$, which must not just be linearly independent among themselves, but also linearly independent from the vectors $\eig^l_k(\kappa)$ in the sequences generated so far. We can continue this process until we reach $m=1$ and need to select $t_1$ vectors in $E^{(1)}(\kappa)$ that must be linearly independent from all the vectors $\eig^l_k(\kappa)$ generated so far. This leaves us with $j_\kappa$ distinct sequences of the form
	\begin{align}
	\eig^1_k(\kappa),\cdots,\eig^{m_k(\kappa)}_k(\kappa)\,,
	\end{align}
	which span the space for the Jordan block $J_k(\kappa)$.
\end{enumerate}
Note that, for complex eigenvalues, the corresponding vectors $\eig^l_k(\kappa)$ are complex. Later, we construct real vectors out of them.\\
\textbf{Step 3:} Having constructed the sequences of generalized eigenvectors, we are ready to construct the different Jordan blocks explicitly. We need to distinguish two cases:
\begin{itemize}
	\item \textbf{Real eigenvalue $\kappa=\lambda\in\mathbb{R}$}\\
	Let us look at the sequence of $v^1_k(\kappa),\cdots, v^{m_k}_k(\kappa)$ of generalized eigenvectors. We can write down the action of $K^\intercal$ on each element and find
		\begin{align}
		K^\intercal\,\eig^l_k(\kappa)&=\lambda\,\eig^l_k(\kappa)+\eig^{l-1}_k(\kappa)\,,\\
		K^\intercal\,\eig^1_k(\kappa)&=\lambda\,\eig^1_k(\kappa)\,,
		\end{align}
	where we have $m_k\geq l>1$ in the first equation. This means that the action of $K^\intercal$ with respect to the generalized eigenvectors of the sequence $\eig^1_k(\kappa),\cdots, \eig^{m_k}_k(\kappa)$ is completely described by the matrix
		\begin{align}
		J_k(\kappa)=\left(\begin{array}{cccc}
			\lambda & 1& &\\[0em]
			& \lambda & \ddots& \\[0em]
			& &\ddots & 1\\[0em]
			& & & \lambda
			\end{array}\right)\,.
		\end{align}
	\item \textbf{Complex eigenvalues $\kappa=\lambda+\ii\omega\in\mathbb{C}$}\\
	For complex eigenvalues $\kappa$, $K^\intercal$ is only diagonalizable over the complex numbers, but not over the reals. However, due to the fact that $K^\intercal$ is a real matrix, we know that all sequences of generalized eigenvectors $\eig^l_k(\kappa)$, which we constructed in the previous step, are actually complex. Their complex conjugated counterparts $\eig^l_k(\kappa)^*$ are themselves generalized eigenvectors associated with the eigenvalue $\kappa^*=\lambda-\ii\omega$:
	\begin{align}
	\qquad K^\intercal\,\eig^l_k(\kappa)^*&=(\lambda-\ii\omega)\,\eig^l_k(\kappa)^*+\eig^{l-1}_k(\kappa)^*\,,\\
	\qquad K^\intercal\,\eig^1_k(\kappa)^*&=(\lambda-\ii\omega)\,\eig^1_k(\kappa)^*\,,
	\end{align}
	where we have $m_k\geq l>1$ in the first equation. We can take linear combinations of generalized eigenvectors and their complex conjugates to find the real vectors
	\begin{align}
		\eig^{l+}_k(\kappa)&=\frac{1}{\sqrt{2}}\left[\eig^l_k(\kappa)+\eig^l_k(\kappa)^*\right]\,,\\
		\eig^{l-}_k(\kappa)&=\frac{\ii}{\sqrt{2}}\left[\eig^l_k(\kappa)-\eig^l_k(\kappa)^*\right]\,.
	\end{align}
	The action of $K$ on these real vectors is then given by
	\begin{align}
	\qquad K^\intercal \eig^{l+}_k(\kappa)&=\lambda \mathcal{E}^{l+}_k(\kappa)+\omega\eig^{l-}_k(\kappa)+\eig^{(l-1)+}_k(\kappa)\,,\\
	\qquad K^\intercal \eig^{l-}_k(\kappa)&=\lambda \eig^{l-}_k(\kappa)-\omega\eig^{l+}_k(\kappa)+\eig^{(l-1)-}_k(\kappa)\,.
	\end{align}
	This means that the action of $K^\intercal$ on the generalized eigenvectors of the sequence $\eig^{1+}_k(\kappa),\eig^{1-}_k(\kappa)\cdots, \eig^{m_k+}_k(\kappa),\eig^{m_k-}_k(\kappa)$ is completely described by the matrix
	\begin{align}
	\quad J_k(\kappa)&=\left(\begin{array}{cccccccc}
	a & b& 1& &&&&\\[0em]
	-b & a & &1&&&&\\[0em]
	&& a & b&\ddots&&&\\[0em]
	&& -b & a &&\ddots&&\\[0em]
	&& && \ddots &&1&\\[0em]
	&&&& & \ddots &&1\\[0em]
	&& &&  &&a&b\\[0em]
	&&&& &  &-b&a\\[0em]
	\end{array}\right)\\
	&=\left(\begin{array}{cccc}
		a \mathds{1}_2+b\mathds{B}_2 & \mathds{1}_2 & &\\[0em]
		& a \mathds{1}_2+b\mathds{B}_2 &\ddots & \\[0em]
		& &\ddots &\mathds{1}_2\\[0em]
		& & & a \mathds{1}_2+b\mathds{B}_2
		\end{array}\right)\,.\nonumber
	\end{align}
\end{itemize}
This is our original claim.
\end{proof}
Bringing a linear map $K^\intercal$ into its Jordan normal form is a generalized diagonalization that can be applied even if $K^\intercal$ is non-diagonalizable. If $K^\intercal$ is diagonalizable over the real numbers, all Jordan blocks are $1\times1$ and the Jordan normal form is just the traditional diagonalization. We can use the construction presented here to decompose every linear map into three parts, which we call the Jordan-Chevalley decomposition.\footnote{Usually, the Jordan-Chevalley decomposition just refers to the decomposition of $K^\intercal$ into a diagonalizable part $K^\intercal_\mathrm{diagonalizable}$ and a nilpotent part $K^\intercal_\mathrm{nilpotent}$. However, for our purposes, it is important to split $K^\intercal_\mathrm{diagonalizable}$ into its real and imaginary parts as well. This is a straightforward generalization of the Jordan-Chevalley decomposition.}
\begin{corollary}[Jordan-Chevalley decomposition]
	The decomposition of each block $J_k(\kappa)$ into the three parts $A_k(\kappa)$, $B_k(\kappa)$, and $C_k(\kappa)$, where the block $B_k(\kappa)$ vanishes for real $\kappa$, induces an overall decomposition of the matrix representation of $K^\intercal$ with respect to our generalized eigenvectors into
	\begin{align}
	K^\intercal= A^\intercal+B^\intercal+C^\intercal\,.
	\end{align}
	Knowing this decomposition in a specific basis allows us to compute the decomposition of $K^\intercal$, or, equivalently, of $K$ into three parts
	\begin{align}
	K=K_\mathrm{real}+K_\mathrm{imaginary}+K_\mathrm{nilpotent}\,,
	\end{align}
	which can be expressed in an arbitrary basis. Here, $K_\mathrm{real}$ is diagonalizable with purely real eigenvalues, $K_\mathrm{imaginary}$ is diagonalizable with purely imaginary eigenvalues, and $K_\mathrm{niloptent}$ is not diagonalizable, but a nilpotent linear map. All three maps commute with each other.
\end{corollary}

\section{Classical time evolution of quadratic Hamiltonians\label{classical}}
Given a quadratic time-independent Hamiltonian $H=\frac{1}{2}h_{ab}\xi^a\xi^b$, the classical equations of motion are given by the Hamilton equations\footnote{We use the symplectic form $\Omega$ whose matrix representation in a canonical basis is $\Omega\equiv\left(\begin{array}{cc}
	0 & \mathds{1}\\[0em]
	-\mathds{1} & 0
	\end{array}\right)$.}
\begin{align}
	\dot{\xi}^a(t)=\Omega^{ab}h_{bc}\xi^c(t)\,.
\end{align}
This is a linear ordinary differential equation whose solution can be completely characterized by the Hamiltonian flow $M(t):V\to V$ with $\xi^a(t)=M(t)^a{}_b\xi^b(0)$. This implies the equation
\begin{align}
	\frac{d}{dt}\,M(t)^a{}_b=\Omega^{ac}h_{cb}\,.
\end{align}
If we define $K^a{}_b=\Omega^{ac}h_{cb}$, the solution $M(t)^a{}_b$ can be written as a matrix exponential:
\begin{align}
	M(t)^a{}_b=\exp\left(t\,K\right)^a{}_b\,.
\end{align}
The matrix $K$ is the infinitesimal generator of the canonical transformation $M(t)$ and satisfies\footnote{In matrix notation, this is the well-known condition $K\Omega+\Omega K^\intercal=0$ for $K$ to be an element of the symplectic algebra $\mathrm{sp}(2N)$. It is equivalent to requiring $K^a{}_b=\Omega^{ac}h_{cb}$ with symmetric $h_{ab}=h_{ba}$.}
\begin{align}
	K^a{}_c\Omega^{cb}+\Omega^{ac}(K^\intercal)_c{}^b=0\,,
\end{align}
where $K^\intercal$ is the transpose of $K$. Therefore, the exponential $M(t)$ preserves $\Omega$ through
\begin{align}
	M(t)^a{}_c\,\Omega^{cd}\,M^\intercal(t)_d{}^b=\Omega^{ab}\,,
\end{align}
which ensures that $M(t)$ is a proper canonical transformation preserving the Poisson bracket.

Studying the time evolution of the entanglement entropy of Gaussian states is equivalent to asking how the classical time evolution $M(t)$ stretches regions of subspaces in classical phase space. For time-independent Hamiltonians, the time evolution $M(t)$ is completely characterized by the Jordan normal form of the symplectic generator $K$. We can use the Jordan-Chevalley decomposition $K=K_\mathrm{real}+K_\mathrm{imaginary}+K_\mathrm{nilpotent}$ to decompose
\begin{align}
	M(t)=\ee^{Kt}=\ee^{K_\mathrm{real}t}\ee^{K_\mathrm{imaginary}t}\ee^{K_\mathrm{nilpotent}t}\,,
\end{align}
where we use the fact that the three parts commute with each other. Let us analyze the action of the three parts on an arbitrary generalized eigenvector $\eig^{l\pm}_k(\kappa)$ for some complex eigenvalue pair $\kappa=\lambda\pm \ii\omega$.
\begin{itemize}
	\item[(a)] \textbf{Real part: Exponential stretching $S(t)\sim\Lambda_A t$}\\
	The real part $K_\mathrm{real}$ has the simplest effect. We find immediately
	\begin{align}
		\ee^{K_\mathrm{real}t}\eig^{l\pm}_k(\kappa)=\ee^{\lambda t}\eig^{l\pm}_k(\kappa)\,,
	\end{align}
	which means the generalized eigenvector is stretched (or squeezed if $\lambda<0$) exponentially with a factor $\ee^{\lambda t}$. For a generic $2N_A$ dimensional region, only the $2N_A$ directions that are stretched the fastest will contribute. In summary, the real part contributes an exponential part $\ee^{\Lambda_A t}$ to the time dependence of a volume, where $\Lambda_A$ is a sum over the largest real parts $\lambda_i$ of the eigenvalues $\kappa_i$. This leads to a linear contribution $\Lambda_A t$ to the entanglement entropy.
	\item[(b)] \textbf{Imaginary part: Rotations $\Rightarrow$ $S_A\sim X_A(t)$}\\
	Let us recall that the imaginary eigenvalues of $K_\mathrm{imaginary}$ always come in conjugate pairs. If we only consider the action of generalized eigenvectors, we find
	\begin{align}
		K_\mathrm{imaginary} \eig^{l+}_k(\kappa)&=\omega\eig^{l-}_k(\kappa)\,,\\
		K_\mathrm{imaginary} \eig^{l-}_k(\kappa)&=-\omega\eig^{l+}_k(\kappa)\,.
	\end{align}
	This implies that the exponentiated action is given by
	\begin{align}\tiny
	&\ee^{K_\mathrm{imaginary}t}
	\left(\begin{array}{c}
	\eig^{l+}_k(\kappa)\\[0em]
	\eig^{l-}_k(\kappa)
	\end{array}\right) \nonumber\\ &=\left(\begin{array}{cc}
	\cos{\omega t} & \sin{\omega t}\\[0em]
	-\sin{\omega t} & \cos{\omega t}
	\end{array}\right)\left(\begin{array}{c}
	\eig^{l+}_k(\kappa)\\[0em]
	\eig^{l-}_k(\kappa)
	\end{array}\right)\,,
	\end{align}
	which corresponds to a rotation in the plane spanned by $\mathcal{E}^{l\pm}_k(\kappa)$. Note that we cannot distinguish a rotation on a circle from an elliptical orbit unless we choose a metric to compute the length of vectors. However, the orbit is always bounded, which means that the imaginary part of $K$ can only change the volume of a region by a constant. A vector that is a general linear combination of many different generalized eigenvectors will follow a complicated, but bounded, trajectory resulting from a superposition of the rotations in different planes with different frequencies $\omega_i =\mathrm{Im}(\kappa_i)$. The imaginary part is responsible for the bounded and oscillating contribution $X_A(t)$ to the entanglement entropy.
	\item[(c)] \textbf{Nilpotent part: Shearing $\Rightarrow$ $S_A(t)\sim C_A\ln{(t)}$}\\
	The nilpotent part $K_\mathrm{nilpotent}$ corresponds to a shearing in the plane spanned by all the vectors of a specific Jordan block $J_k(\kappa)$. The action of the generator
	\begin{align}
	\qquad K_\mathrm{nilpotent}\eig^{l\pm}_k(\kappa)=\eig^{(l-1)\pm}_k(\kappa)\,,\quad \eig^{0}_k(\kappa)=0,
	\end{align}
	exponentiates to the action
	\begin{align}
	\qquad \ee^{K_\mathrm{nilpotent}t}\,\eig^{l\pm}_k(\kappa)=\sum^l_{l'=1}\frac{t^{l-l'}}{(l-l')!}\eig^{l'\pm}_k(\kappa)\,.
	\end{align}
	The volume of a region spanned by \emph{all} vector $\eig^{l'\pm}_k(\kappa)$ with $1\leq l'\leq l$ does not change under this time evolution, but for a generic region that is only stretched along some directions of this subspace, there will be a polynomial stretching. The largest exponent of a single direction is clearly given by $\dim J_k(\kappa)-1$. If we need to choose $n_k(\kappa)$ vectors in the Jordan block $J_k(\kappa)$, the maximal exponent is given by
	\begin{align}
		\left[\dim J_k(\kappa)-n_k(\kappa)\right]\,n_k(\kappa)\,.
	\end{align}
	In summary, the nilpotent part only contributes a polynomial growth with integer exponents ${C_A}$ to the time dependence of a volume. After taking the logarithm, we have a contribution $C_A\ln(t)$ to the entanglement entropy.
\end{itemize}

\section{Numerical calculations\label{numericalmethods}}
In a lattice, unlike for fermions, bosons have an infinite local dimension of the Hilbert space. Hence, even the two-site Hamiltonian~\eqref{twosite_hamiltonian} cannot be diagonalized exactly. In order to do numerical calculations, we truncate the Hilbert space to a $(T+1)$-dimensional subsector given by Eq.~\eqref{span}. We use full exact diagonalization within the truncated basis to obtain the entanglement entropy in all cases except for $\Delta=1$ and $U=0$ (for which the entanglement production is logarithmic). In what follows, we explain the numerical technique used to obtain the entanglement entropy for $\Delta=1$ and $U=0$, which is more efficient than full exact diagonalization. 

A many body state of the Hamiltonian~\eqref{twosite_hamiltonian} can be written as $|\psi(t)\rangle=\sum_{n}C_{n}(t) |n,n\rangle$. The elements of the reduced density matrix of site $1$, which is diagonal with diagonal matrix elements $\rho_{n,n}^{1}(t)=C_{n}(t)C_{n}^{*}(t)$, can be computed using that
\begin{equation}
f_k(t)=\langle\psi(t)|(\hat{a}_{1}^{\dagger}\hat{a}^{}_1)^k|\psi(t)\rangle
=\sum_{n}n^{k}\rho_{n,n}^{1}(t),
\label{Eq: reduced matrix}
\end{equation}
so that
\begin{equation}
\rho_{n,n}^{1}(t)=\sum_kV^{-1}_{nk}f_{k}(t),\label{Eq:inverse}
\end{equation}
where $V$ is a Vandermonde matrix, whose inverse can be calculated analytically~\cite{Knuth.1997}.

The computation of $f_{k}(t)$ is done in the Heisenberg picture,
\begin{equation}
f_k(t)=\langle\psi(0)|\left[\hat{a}_{1}^{\dagger}(t)\hat{a}^{}_{1}(t)\right]^k|\psi(0)\rangle,
\label{rd}
\end{equation}
by noticing that the Hamiltonian~\eqref{twosite_hamiltonian} can be written in a diagonal form using the operators: $\hat{a}^{\dagger}_{e} = (\hat{a}^{\dagger}_{1}+\hat{a}^{\dagger}_{2}) / \sqrt{2}$ and $\hat{a}^{\dagger}_{o} = (\hat{a}^{\dagger}_{1}-\hat{a}^{\dagger}_{2}) / \sqrt{2}$, which obey standard bosonic commutation relations. In terms of these operators, the Hamiltonian reads
\begin{equation}
\hat{H}(\Delta,0) = \hat{a}_{e}^{\dagger}\hat{a}^{}_{e} + \hat{a}_{o}^{\dagger}\hat{a}^{}_{o} + \frac{\Delta}{2} (\hat{a}_{e}^{\dagger}\hat{a}_e^{\dagger} - \hat{a}_{o}^{\dagger}\hat{a}_{o}^{\dagger} + \text{H.c.}).
\end{equation}
The Heisenberg equations of motion for $\hat{a}^{\dagger}_{e}$ and $\hat{a}^{\dagger}_{o}$ are
\begin{equation}
\begin{aligned}
\dot{\hat{a}}_e^{\dagger}&= \ii[\hat{H},\hat{a}_e^{\dagger}]= \ii \hat{a}_e +\ii \hat{a}_e^{\dagger}, \\
\dot{\hat{a}}_o^{\dagger}&=\ii[\hat{H},\hat{a}_o^{\dagger}]=-\ii \hat{a}_o +\ii \hat{a}_o^{\dagger}, \label{p1}
\end{aligned}
\end{equation}
from which it follows that
\begin{equation}
\begin{aligned}
\dot{\hat{a}}_e^{\dagger}+\dot{\hat{a}}_e&= 0, \\[0em]
\dot{\hat{a}}_e^{\dagger}-\dot{\hat{a}}_e&=2\ii(\hat{a}_e^{\dagger}+\hat{a}_e). 
\end{aligned}
\end{equation}
Hence
\begin{equation}
\begin{aligned}
\hat{a}_e^{\dagger}(t)+\hat{a}_e(t)&= \hat{a}_e^{\dagger}(0)+\hat{a}_e(0), \label{Eq:tdo1} \\[0em]
\hat{a}_e^{\dagger}(t)-\hat{a}_e(t)&= 2\ii t\left[\hat{a}_e^{\dagger}(0)+\hat{a}_e(0)\right]+ \hat{a}_e^{\dagger}(0)-\hat{a}_e(0).
\end{aligned}
\end{equation}
Similarly
\begin{equation}
\begin{aligned}
\hat{a}_o^{\dagger}(t)-\hat{a}_o(t)&= \hat{a}_o^{\dagger}(0)-\hat{a}_o(0), \label{Eq:tdo2} \\
\hat{a}_o^{\dagger}(t)+\hat{a}_o(t)&= 2\ii t\left[\hat{a}_o^{\dagger}(0)-\hat{a}_o(0)\right]+ \hat{a}_o^{\dagger}(0)+\hat{a}_o(0). \end{aligned}
\end{equation}
From Eqs.~\eqref{Eq:tdo1} and~\eqref{Eq:tdo2}, one gets that
\begin{equation}
\begin{aligned}
\hat{a}_1^{\dagger}(t)&=(1+\ii t)\hat{a}_{1}^{\dagger}(0)+\ii t\hat{a}_{2}(0), \\
\hat{a}_{2}^{\dagger}(t)&=(1+\ii t)\hat{a}_{2}^{\dagger}(0)+\ii t\hat{a}_{1}(0).
\label{teo}
\end{aligned}
\end{equation}

We evaluate $f_k(t)$ by substituting the results from Eq.~\eqref{teo} in Eq.~\eqref{rd}. For large $T$ (truncated Hilbert space dimension), the computation of the reduced density matrix using Eq.~\eqref{Eq:inverse} involves the addition and multiplication of very large numbers. For $\Delta=1$, all those calculations can be done using integers, which allows us to obtain results with the desired numerical accuracy. In general, for $\Delta \neq 1$, the expressions for $\hat{a}_{1}^{\dagger}(t)$ and $\hat{a}^{\dagger}_{2}(t)$ contain irrational numbers, as a result of which numerical errors render this approach ineffective. 

\bibliographystyle{biblev1}
\bibliography{ref}

\begin{thebibliography}{10}
\expandafter\ifx\csname url\endcsname\relax
  \def\url#1{{\tt #1}}\fi
\expandafter\ifx\csname urlprefix\endcsname\relax\def\urlprefix{URL }\fi
\expandafter\ifx\csname bibinfo\endcsname\relax\def\bibinfo#1#2{#2}\fi
\expandafter\ifx\csname eprint\endcsname\relax\def\eprint#1{\url{#1}}\fi

\bibitem{Calabrese:2005in}
\bibinfo{author}{P.~Calabrese} and \bibinfo{author}{J.~L. Cardy},
  \bibinfo{title}{{Evolution of entanglement entropy in one-dimensional
  systems}},
  \bibinfo{journal}{\href{http://dx.doi.org/10.1088/1742-5468/2005/04/P04010}{J.
  Stat. Mech.}}
  \href{http://dx.doi.org/10.1088/1742-5468/2005/04/P04010}{\bibinfo{pages}{P04010}}
  (\href{http://dx.doi.org/10.1088/1742-5468/2005/04/P04010}{\bibinfo{year}{2005}}).

\bibitem{calabrese2007quantum}
\bibinfo{author}{P.~Calabrese} and \bibinfo{author}{J.~Cardy},
  \bibinfo{title}{Quantum quenches in extended systems},
  \bibinfo{journal}{\href{http://dx.doi.org/10.1088/1742-5468/2007/06/P06008}{J.
  Stat. Mech.}}
  \href{http://dx.doi.org/10.1088/1742-5468/2007/06/P06008}{\bibinfo{pages}{P06008}}
  (\href{http://dx.doi.org/10.1088/1742-5468/2007/06/P06008}{\bibinfo{year}{2007}}).

\bibitem{cotler2016entanglement}
\bibinfo{author}{J.~S. Cotler}, \bibinfo{author}{M.~P. Hertzberg},
  \bibinfo{author}{M.~Mezei}, and \bibinfo{author}{M.~T. Mueller},
  \bibinfo{title}{Entanglement growth after a global quench in free scalar
  field theory},
  \bibinfo{journal}{\href{http://dx.doi.org/10.1007/JHEP11(2016)166}{JHEP}}
  \href{http://dx.doi.org/10.1007/JHEP11(2016)166}{{\bf
  \bibinfo{volume}{2016}}, \bibinfo{pages}{166}}
  (\href{http://dx.doi.org/10.1007/JHEP11(2016)166}{\bibinfo{year}{2016}}).

\bibitem{Hubeny:2007xt}
\bibinfo{author}{V.~E. Hubeny}, \bibinfo{author}{M.~Rangamani}, and
  \bibinfo{author}{T.~Takayanagi}, \bibinfo{title}{{A Covariant holographic
  entanglement entropy proposal}},
  \bibinfo{journal}{\href{http://dx.doi.org/10.1088/1126-6708/2007/07/062}{JHEP}}
  \href{http://dx.doi.org/10.1088/1126-6708/2007/07/062}{{\bf
  \bibinfo{volume}{07}}, \bibinfo{pages}{062}}
  (\href{http://dx.doi.org/10.1088/1126-6708/2007/07/062}{\bibinfo{year}{2007}}).

\bibitem{AbajoArrastia:2010yt}
\bibinfo{author}{J.~Abajo-Arrastia}, \bibinfo{author}{J.~Aparicio}, and
  \bibinfo{author}{E.~Lopez}, \bibinfo{title}{{Holographic Evolution of
  Entanglement Entropy}},
  \bibinfo{journal}{\href{http://dx.doi.org/10.1007/JHEP11(2010)149}{JHEP}}
  \href{http://dx.doi.org/10.1007/JHEP11(2010)149}{{\bf
  \bibinfo{volume}{1011}}, \bibinfo{pages}{149}}
  (\href{http://dx.doi.org/10.1007/JHEP11(2010)149}{\bibinfo{year}{2010}}).

\bibitem{Hartman:2013qma}
\bibinfo{author}{T.~Hartman} and \bibinfo{author}{J.~Maldacena},
  \bibinfo{title}{{Time Evolution of Entanglement Entropy from Black Hole
  Interiors}},
  \bibinfo{journal}{\href{http://dx.doi.org/10.1007/JHEP05(2013)014}{JHEP}}
  \href{http://dx.doi.org/10.1007/JHEP05(2013)014}{{\bf
  \bibinfo{volume}{1305}}, \bibinfo{pages}{014}}
  (\href{http://dx.doi.org/10.1007/JHEP05(2013)014}{\bibinfo{year}{2013}}).

\bibitem{Liu:2013qca}
\bibinfo{author}{H.~Liu} and \bibinfo{author}{S.~J. Suh},
  \bibinfo{title}{{Entanglement growth during thermalization in holographic
  systems}},
  \bibinfo{journal}{\href{http://dx.doi.org/10.1103/PhysRevD.89.066012}{Phys.
  Rev. D}} \href{http://dx.doi.org/10.1103/PhysRevD.89.066012}{{\bf
  \bibinfo{volume}{89}}, \bibinfo{pages}{066012}}
  (\href{http://dx.doi.org/10.1103/PhysRevD.89.066012}{\bibinfo{year}{2014}}).

\bibitem{Bianchi:2014bma}
\bibinfo{author}{E.~Bianchi}, \bibinfo{author}{T.~De~Lorenzo}, and
  \bibinfo{author}{M.~Smerlak}, \bibinfo{title}{{Entanglement entropy
  production in gravitational collapse: covariant regularization and solvable
  models}},
  \bibinfo{journal}{\href{http://dx.doi.org/10.1007/JHEP06(2015)180}{JHEP}}
  \href{http://dx.doi.org/10.1007/JHEP06(2015)180}{{\bf \bibinfo{volume}{06}},
  \bibinfo{pages}{180}}
  (\href{http://dx.doi.org/10.1007/JHEP06(2015)180}{\bibinfo{year}{2015}}).

\bibitem{dechiara2006entanglement}
\bibinfo{author}{G.~DeChiara}, \bibinfo{author}{S.~Montangero},
  \bibinfo{author}{P.~Calabrese}, and \bibinfo{author}{R.~Fazio},
  \bibinfo{title}{Entanglement entropy dynamics of heisenberg chains},
  \bibinfo{journal}{\href{http://dx.doi.org/10.1088/1742-5468/2006/03/P03001}{J.
  Stat. Mech.}} \href{http://dx.doi.org/10.1088/1742-5468/2006/03/P03001}{{\bf
  \bibinfo{volume}{3}}, \bibinfo{pages}{03001}}
  (\href{http://dx.doi.org/10.1088/1742-5468/2006/03/P03001}{\bibinfo{year}{2006}}).

\bibitem{fagotti2008evolution}
\bibinfo{author}{M.~Fagotti} and \bibinfo{author}{P.~Calabrese},
  \bibinfo{title}{Evolution of entanglement entropy following a quantum quench:
  Analytic results for the {XY} chain in a transverse magnetic field},
  \bibinfo{journal}{\href{http://dx.doi.org/10.1103/PhysRevA.78.010306}{Phys.
  Rev. A}} \href{http://dx.doi.org/10.1103/PhysRevA.78.010306}{{\bf
  \bibinfo{volume}{78}}, \bibinfo{pages}{010306}}
  (\href{http://dx.doi.org/10.1103/PhysRevA.78.010306}{\bibinfo{year}{2008}}).

\bibitem{eisler2008entanglement}
\bibinfo{author}{V.~Eisler} and \bibinfo{author}{I.~Peschel},
  \bibinfo{title}{Entanglement in a periodic quench},
  \bibinfo{journal}{\href{http://dx.doi.org/10.1002/andp.200810299}{Annalen der
  Physik}} \href{http://dx.doi.org/10.1002/andp.200810299}{{\bf
  \bibinfo{volume}{17}}, \bibinfo{pages}{410}}
  (\href{http://dx.doi.org/10.1002/andp.200810299}{\bibinfo{year}{2008}}).

\bibitem{lauchli2008spreading}
\bibinfo{author}{A.~M. L{\"a}uchli} and \bibinfo{author}{C.~Kollath},
  \bibinfo{title}{Spreading of correlations and entanglement after a quench in
  the one-dimensional {B}ose--{H}ubbard model},
  \bibinfo{journal}{\href{http://dx.doi.org/10.1088/1742-5468/2008/05/P05018}{J.
  Stat. Mech.}} \href{http://dx.doi.org/10.1088/1742-5468/2008/05/P05018}{{\bf
  \bibinfo{volume}{2008}}, \bibinfo{pages}{P05018}}
  (\href{http://dx.doi.org/10.1088/1742-5468/2008/05/P05018}{\bibinfo{year}{2008}}).

\bibitem{kim2013ballistic}
\bibinfo{author}{H.~Kim} and \bibinfo{author}{D.~A. Huse},
  \bibinfo{title}{Ballistic spreading of entanglement in a diffusive
  nonintegrable system},
  \bibinfo{journal}{\href{http://dx.doi.org/10.1103/PhysRevLett.111.127205}{Phys.
  Rev. Lett.}} \href{http://dx.doi.org/10.1103/PhysRevLett.111.127205}{{\bf
  \bibinfo{volume}{111}}, \bibinfo{pages}{127205}}
  (\href{http://dx.doi.org/10.1103/PhysRevLett.111.127205}{\bibinfo{year}{2013}}).

\bibitem{alba2017entanglement}
\bibinfo{author}{V.~Alba} and \bibinfo{author}{P.~Calabrese},
  \bibinfo{title}{Entanglement and thermodynamics after a quantum quench in
  integrable systems},
  \bibinfo{journal}{\href{http://dx.doi.org/10.1073/pnas.1703516114}{Proc.
  Natl. Acad. Sci.}} \href{http://dx.doi.org/10.1073/pnas.1703516114}{{\bf
  \bibinfo{volume}{114}}, \bibinfo{pages}{7947}}
  (\href{http://dx.doi.org/10.1073/pnas.1703516114}{\bibinfo{year}{2017}}).

\bibitem{giovannetti2014ultimate}
\bibinfo{author}{V.~Giovannetti}, \bibinfo{author}{R.~Garcia-Patron},
  \bibinfo{author}{N.~J. Cerf}, and \bibinfo{author}{A.~S. Holevo},
  \bibinfo{title}{Ultimate classical communication rates of quantum optical
  channels},
  \bibinfo{journal}{\href{http://dx.doi.org/10.1038/nphoton.2014.216}{Nature
  Photonics}} \href{http://dx.doi.org/10.1038/nphoton.2014.216}{{\bf
  \bibinfo{volume}{8}}, \bibinfo{pages}{796}}
  (\href{http://dx.doi.org/10.1038/nphoton.2014.216}{\bibinfo{year}{2014}}).

\bibitem{gagatsos2016entropy}
\bibinfo{author}{C.~N. Gagatsos}, \bibinfo{author}{A.~I. Karanikas},
  \bibinfo{author}{G.~Kordas}, and \bibinfo{author}{N.~J. Cerf},
  \bibinfo{title}{Entropy generation in gaussian quantum transformations:
  applying the replica method to continuous-variable quantum information
  theory}, \bibinfo{journal}{\href{http://dx.doi.org/10.1038/npjqi.2015.8}{npj
  Quantum Information}} \href{http://dx.doi.org/10.1038/npjqi.2015.8}{{\bf
  \bibinfo{volume}{2}}, \bibinfo{pages}{15008}}
  (\href{http://dx.doi.org/10.1038/npjqi.2015.8}{\bibinfo{year}{2016}}).

\bibitem{de2017gaussian}
\bibinfo{author}{G.~De~Palma}, \bibinfo{author}{D.~Trevisan}, and
  \bibinfo{author}{V.~Giovannetti}, \bibinfo{title}{Gaussian states minimize
  the output entropy of one-mode quantum gaussian channels},
  \bibinfo{journal}{\href{http://dx.doi.org/10.1103/PhysRevLett.118.160503}{Phys.
  Rev. Lett.}} \href{http://dx.doi.org/10.1103/PhysRevLett.118.160503}{{\bf
  \bibinfo{volume}{118}}, \bibinfo{pages}{160503}}
  (\href{http://dx.doi.org/10.1103/PhysRevLett.118.160503}{\bibinfo{year}{2017}}).

\bibitem{islam_ma_15}
\bibinfo{author}{R.~Islam}, \bibinfo{author}{R.~Ma}, \bibinfo{author}{P.~M.
  Preiss}, \bibinfo{author}{M.~E. Tai}, \bibinfo{author}{A.~Lukin},
  \bibinfo{author}{M.~Rispoli}, and \bibinfo{author}{M.~Greiner},
  \bibinfo{title}{Measuring entanglement entropy in a quantum many-body
  system},
  \bibinfo{journal}{\href{http://dx.doi.org/10.1038/nature15750}{Nature}}
  \href{http://dx.doi.org/10.1038/nature15750}{{\bf \bibinfo{volume}{528}},
  \bibinfo{pages}{77}}
  (\href{http://dx.doi.org/10.1038/nature15750}{\bibinfo{year}{2015}}).

\bibitem{kaufman_tai_16}
\bibinfo{author}{A.~M. Kaufman}, \bibinfo{author}{M.~E. Tai},
  \bibinfo{author}{A.~Lukin}, \bibinfo{author}{M.~Rispoli},
  \bibinfo{author}{R.~Schittko}, \bibinfo{author}{P.~M. Preiss}, and
  \bibinfo{author}{M.~Greiner}, \bibinfo{title}{Quantum thermalization through
  entanglement in an isolated many-body system},
  \bibinfo{journal}{\href{http://dx.doi.org/10.1126/science.aaf6725}{Science}}
  \href{http://dx.doi.org/10.1126/science.aaf6725}{{\bf \bibinfo{volume}{353}},
  \bibinfo{pages}{794}}
  (\href{http://dx.doi.org/10.1126/science.aaf6725}{\bibinfo{year}{2016}}).

\bibitem{bianchi17}
\bibinfo{author}{E.~Bianchi}, \bibinfo{author}{L.~Hackl}, and
  \bibinfo{author}{N.~Yokomizo}, \bibinfo{title}{Linear growth of the
  entanglement entropy and the {Kolmogorov-Sinai} rate},
  \href{https://arxiv.org/abs/1709.00427}{\bibinfo{howpublished}{arXiv:1709.00427}}.

\bibitem{nandkishore_huse_review_15}
\bibinfo{author}{R.~Nandkishore} and \bibinfo{author}{D.~A. Huse},
  \bibinfo{title}{Many-body localization and thermalization in quantum
  statistical mechanics},
  \bibinfo{journal}{\href{http://dx.doi.org/10.1146/annurev-conmatphys-031214-014726}{Annual
  Review of Condensed Matter Physics}}
  \href{http://dx.doi.org/10.1146/annurev-conmatphys-031214-014726}{{\bf
  \bibinfo{volume}{6}}, \bibinfo{pages}{15}}
  (\href{http://dx.doi.org/10.1146/annurev-conmatphys-031214-014726}{\bibinfo{year}{2015}}).

\bibitem{dalessio_kafri_16}
\bibinfo{author}{L.~D'Alessio}, \bibinfo{author}{Y.~Kafri},
  \bibinfo{author}{A.~Polkovnikov}, and \bibinfo{author}{M.~Rigol},
  \bibinfo{title}{From quantum chaos and eigenstate thermalization to
  statistical mechanics and thermodynamics},
  \bibinfo{journal}{\href{http://dx.doi.org/10.1080/00018732.2016.1198134}{Adv.
  Phys.}} \href{http://dx.doi.org/10.1080/00018732.2016.1198134}{{\bf
  \bibinfo{volume}{65}}, \bibinfo{pages}{239}}
  (\href{http://dx.doi.org/10.1080/00018732.2016.1198134}{\bibinfo{year}{2016}}).

\bibitem{mukhanov2005physical}
\bibinfo{author}{V.~Mukhanov}, {\em \bibinfo{title}{Physical foundations of
  cosmology}\/} (\bibinfo{publisher}{Cambridge University Press},
  \bibinfo{year}{2005}).

\bibitem{weinberg2008cosmology}
\bibinfo{author}{S.~Weinberg}, {\em \bibinfo{title}{Cosmology}\/}
  (\bibinfo{publisher}{Oxford University Press}, \bibinfo{year}{2008}).

\bibitem{Campo:2005sy}
\bibinfo{author}{D.~Campo} and \bibinfo{author}{R.~Parentani},
  \bibinfo{title}{{Inflationary spectra and partially decohered
  distributions}},
  \bibinfo{journal}{\href{http://dx.doi.org/10.1103/PhysRevD.72.045015}{Phys.
  Rev. D}} \href{http://dx.doi.org/10.1103/PhysRevD.72.045015}{{\bf
  \bibinfo{volume}{72}}, \bibinfo{pages}{045015}}
  (\href{http://dx.doi.org/10.1103/PhysRevD.72.045015}{\bibinfo{year}{2005}}).

\bibitem{Polarski:1995jg}
\bibinfo{author}{D.~Polarski} and \bibinfo{author}{A.~A. Starobinsky},
  \bibinfo{title}{{Semiclassicality and decoherence of cosmological
  perturbations}},
  \bibinfo{journal}{\href{http://dx.doi.org/10.1088/0264-9381/13/3/006}{Class.
  Quant. Grav.}} \href{http://dx.doi.org/10.1088/0264-9381/13/3/006}{{\bf
  \bibinfo{volume}{13}}, \bibinfo{pages}{377}}
  (\href{http://dx.doi.org/10.1088/0264-9381/13/3/006}{\bibinfo{year}{1996}}).

\bibitem{Kiefer:1999sj}
\bibinfo{author}{C.~Kiefer}, \bibinfo{author}{D.~Polarski}, and
  \bibinfo{author}{A.~A. Starobinsky}, \bibinfo{title}{{Entropy of gravitons
  produced in the early universe}},
  \bibinfo{journal}{\href{http://dx.doi.org/10.1103/PhysRevD.62.043518}{Phys.
  Rev. D}} \href{http://dx.doi.org/10.1103/PhysRevD.62.043518}{{\bf
  \bibinfo{volume}{62}}, \bibinfo{pages}{043518}}
  (\href{http://dx.doi.org/10.1103/PhysRevD.62.043518}{\bibinfo{year}{2000}}).

\bibitem{Martin:2015qta}
\bibinfo{author}{J.~Martin} and \bibinfo{author}{V.~Vennin},
  \bibinfo{title}{{Quantum Discord of Cosmic Inflation: Can we Show that CMB
  Anisotropies are of Quantum-Mechanical Origin?}},
  \bibinfo{journal}{\href{http://dx.doi.org/10.1103/PhysRevD.93.023505}{Phys.
  Rev. D}} \href{http://dx.doi.org/10.1103/PhysRevD.93.023505}{{\bf
  \bibinfo{volume}{93}}, \bibinfo{pages}{023505}}
  (\href{http://dx.doi.org/10.1103/PhysRevD.93.023505}{\bibinfo{year}{2016}}).

\bibitem{traschen1990particle}
\bibinfo{author}{J.~H. Traschen} and \bibinfo{author}{R.~H. Brandenberger},
  \bibinfo{title}{Particle production during out-of-equilibrium phase
  transitions},
  \bibinfo{journal}{\href{http://dx.doi.org/10.1103/PhysRevD.42.2491}{Phys.
  Rev. D}} \href{http://dx.doi.org/10.1103/PhysRevD.42.2491}{{\bf
  \bibinfo{volume}{42}}, \bibinfo{pages}{2491}}
  (\href{http://dx.doi.org/10.1103/PhysRevD.42.2491}{\bibinfo{year}{1990}}).

\bibitem{kofman1994reheating}
\bibinfo{author}{L.~Kofman}, \bibinfo{author}{A.~Linde}, and
  \bibinfo{author}{A.~A. Starobinsky}, \bibinfo{title}{Reheating after
  inflation},
  \bibinfo{journal}{\href{http://dx.doi.org/10.1103/PhysRevLett.73.3195}{Phys.
  Rev. Lett.}} \href{http://dx.doi.org/10.1103/PhysRevLett.73.3195}{{\bf
  \bibinfo{volume}{73}}, \bibinfo{pages}{3195}}
  (\href{http://dx.doi.org/10.1103/PhysRevLett.73.3195}{\bibinfo{year}{1994}}).

\bibitem{allahverdi2010reheating}
\bibinfo{author}{R.~Allahverdi}, \bibinfo{author}{R.~Brandenberger},
  \bibinfo{author}{F.-Y. Cyr-Racine}, and \bibinfo{author}{A.~Mazumdar},
  \bibinfo{title}{Reheating in inflationary cosmology: theory and
  applications},
  \bibinfo{journal}{\href{http://dx.doi.org/10.1146/annurev.nucl.012809.104511}{Annu.
  Rev. Nucl. Part. Sci.}}
  \href{http://dx.doi.org/10.1146/annurev.nucl.012809.104511}{{\bf
  \bibinfo{volume}{60}}, \bibinfo{pages}{27}}
  (\href{http://dx.doi.org/10.1146/annurev.nucl.012809.104511}{\bibinfo{year}{2010}}).

\bibitem{amin2015nonperturbative}
\bibinfo{author}{M.~A. Amin}, \bibinfo{author}{M.~P. Hertzberg},
  \bibinfo{author}{D.~I. Kaiser}, and \bibinfo{author}{J.~Karouby},
  \bibinfo{title}{Nonperturbative dynamics of reheating after inflation: a
  review},
  \bibinfo{journal}{\href{http://dx.doi.org/10.1142/S0218271815300037}{‎Int.
  J. Mod. Phys. D}} \href{http://dx.doi.org/10.1142/S0218271815300037}{{\bf
  \bibinfo{volume}{24}}, \bibinfo{pages}{1530003}}
  (\href{http://dx.doi.org/10.1142/S0218271815300037}{\bibinfo{year}{2015}}).

\bibitem{yagi2005quark}
\bibinfo{author}{K.~Yagi}, \bibinfo{author}{T.~Hatsuda}, and
  \bibinfo{author}{Y.~Miake}, {\em \bibinfo{title}{Quark-gluon plasma: From big
  bang to little bang}\/}, volume~\bibinfo{volume}{23}
  (\bibinfo{publisher}{Cambridge University Press}, \bibinfo{year}{2005}).

\bibitem{muller2011entropy}
\bibinfo{author}{B.~M{\"u}ller} and \bibinfo{author}{A.~Sch{\"a}fer},
  \bibinfo{title}{Entropy creation in relativistic heavy ion collisions},
  \bibinfo{journal}{\href{http://dx.doi.org/10.1142/S0218301311020459}{‎Int.
  J. Mod. Phys. E}} \href{http://dx.doi.org/10.1142/S0218301311020459}{{\bf
  \bibinfo{volume}{20}}, \bibinfo{pages}{2235}}
  (\href{http://dx.doi.org/10.1142/S0218301311020459}{\bibinfo{year}{2011}}).

\bibitem{kunihiro2010chaotic}
\bibinfo{author}{T.~Kunihiro}, \bibinfo{author}{B.~M{\"u}ller},
  \bibinfo{author}{A.~Ohnishi}, \bibinfo{author}{A.~Sch{\"a}fer},
  \bibinfo{author}{T.~T. Takahashi}, and \bibinfo{author}{A.~Yamamoto},
  \bibinfo{title}{Chaotic behavior in classical yang-mills dynamics},
  \bibinfo{journal}{\href{http://dx.doi.org/10.1103/PhysRevD.82.114015}{Phys.
  Rev. D}} \href{http://dx.doi.org/10.1103/PhysRevD.82.114015}{{\bf
  \bibinfo{volume}{82}}, \bibinfo{pages}{114015}}
  (\href{http://dx.doi.org/10.1103/PhysRevD.82.114015}{\bibinfo{year}{2010}}).

\bibitem{Hashimoto:2016wme}
\bibinfo{author}{K.~Hashimoto}, \bibinfo{author}{K.~Murata}, and
  \bibinfo{author}{K.~Yoshida}, \bibinfo{title}{{Chaos in chiral condensates in
  gauge theories}},
  \bibinfo{journal}{\href{http://dx.doi.org/10.1103/PhysRevLett.117.231602}{Phys.
  Rev. Lett.}} \href{http://dx.doi.org/10.1103/PhysRevLett.117.231602}{{\bf
  \bibinfo{volume}{117}}, \bibinfo{pages}{231602}}
  (\href{http://dx.doi.org/10.1103/PhysRevLett.117.231602}{\bibinfo{year}{2016}}).

\bibitem{fedichev2004cosmological}
\bibinfo{author}{P.~O. Fedichev} and \bibinfo{author}{U.~R. Fischer},
  \bibinfo{title}{Cosmological quasiparticle production in harmonically trapped
  superfluid gases},
  \bibinfo{journal}{\href{http://dx.doi.org/10.1103/PhysRevA.69.033602}{Phys.
  Rev. A}} \href{http://dx.doi.org/10.1103/PhysRevA.69.033602}{{\bf
  \bibinfo{volume}{69}}, \bibinfo{pages}{033602}}
  (\href{http://dx.doi.org/10.1103/PhysRevA.69.033602}{\bibinfo{year}{2004}}).

\bibitem{carusotto2010density}
\bibinfo{author}{I.~Carusotto}, \bibinfo{author}{R.~Balbinot},
  \bibinfo{author}{A.~Fabbri}, and \bibinfo{author}{A.~Recati},
  \bibinfo{title}{Density correlations and analog dynamical casimir emission of
  bogoliubov phonons in modulated atomic {B}ose-{E}instein condensates},
  \bibinfo{journal}{\href{http://dx.doi.org/10.1140/epjd/e2009-00314-3}{Eur.
  Phys. J. D}} \href{http://dx.doi.org/10.1140/epjd/e2009-00314-3}{{\bf
  \bibinfo{volume}{56}}, \bibinfo{pages}{391}}
  (\href{http://dx.doi.org/10.1140/epjd/e2009-00314-3}{\bibinfo{year}{2010}}).

\bibitem{jaskula2012acoustic}
\bibinfo{author}{J.-C. Jaskula}, \bibinfo{author}{G.~B. Partridge},
  \bibinfo{author}{M.~Bonneau}, \bibinfo{author}{R.~Lopes},
  \bibinfo{author}{J.~Ruaudel}, \bibinfo{author}{D.~Boiron}, and
  \bibinfo{author}{C.~I. Westbrook}, \bibinfo{title}{Acoustic analog to the
  dynamical casimir effect in a {Bose-Einstein} condensate},
  \bibinfo{journal}{\href{http://dx.doi.org/10.1103/PhysRevLett.109.220401}{Phys.
  Rev. Lett.}} \href{http://dx.doi.org/10.1103/PhysRevLett.109.220401}{{\bf
  \bibinfo{volume}{109}}, \bibinfo{pages}{220401}}
  (\href{http://dx.doi.org/10.1103/PhysRevLett.109.220401}{\bibinfo{year}{2012}}).

\bibitem{Steinhauer:2015saa}
\bibinfo{author}{J.~Steinhauer}, \bibinfo{title}{{Observation of quantum
  Hawking radiation and its entanglement in an analogue black hole}},
  \bibinfo{journal}{\href{http://dx.doi.org/10.1038/nphys3863}{Nature Phys.}}
  \href{http://dx.doi.org/10.1038/nphys3863}{{\bf \bibinfo{volume}{12}},
  \bibinfo{pages}{959}}
  (\href{http://dx.doi.org/10.1038/nphys3863}{\bibinfo{year}{2016}}).

\bibitem{holevo2013quantum}
\bibinfo{author}{A.~S. Holevo}, {\em \bibinfo{title}{Quantum systems, channels,
  information: a mathematical introduction}\/}, volume~\bibinfo{volume}{16}
  (\bibinfo{publisher}{Walter de Gruyter}, \bibinfo{year}{2013}).

\bibitem{weedbrook2012gaussian}
\bibinfo{author}{C.~Weedbrook}, \bibinfo{author}{S.~Pirandola},
  \bibinfo{author}{R.~Garcia-Patron}, \bibinfo{author}{N.~J. Cerf},
  \bibinfo{author}{T.~C. Ralph}, \bibinfo{author}{J.~H. Shapiro}, and
  \bibinfo{author}{S.~Lloyd}, \bibinfo{title}{Gaussian quantum information},
  \bibinfo{journal}{\href{http://dx.doi.org/10.1103/RevModPhys.84.621}{Rev.
  Mod. Phys.}} \href{http://dx.doi.org/10.1103/RevModPhys.84.621}{{\bf
  \bibinfo{volume}{84}}, \bibinfo{pages}{621}}
  (\href{http://dx.doi.org/10.1103/RevModPhys.84.621}{\bibinfo{year}{2012}}).

\bibitem{adesso2014continuous}
\bibinfo{author}{G.~Adesso}, \bibinfo{author}{S.~Ragy}, and
  \bibinfo{author}{A.~R. Lee}, \bibinfo{title}{Continuous variable quantum
  information: Gaussian states and beyond}, \bibinfo{journal}{Open Systems \&
  Information Dynamics} {\bf \bibinfo{volume}{21}}, \bibinfo{pages}{1440001}
  (\bibinfo{year}{2014}).

\bibitem{Bianchi:2015fra}
\bibinfo{author}{E.~Bianchi}, \bibinfo{author}{L.~Hackl}, and
  \bibinfo{author}{N.~Yokomizo}, \bibinfo{title}{{Entanglement entropy of
  squeezed vacua on a lattice}},
  \bibinfo{journal}{\href{http://dx.doi.org/10.1103/PhysRevD.92.085045}{Phys.
  Rev. D}} \href{http://dx.doi.org/10.1103/PhysRevD.92.085045}{{\bf
  \bibinfo{volume}{92}}, \bibinfo{pages}{085045}}
  (\href{http://dx.doi.org/10.1103/PhysRevD.92.085045}{\bibinfo{year}{2015}}).

\bibitem{Bianchi2017kahler}
\bibinfo{author}{E.~Bianchi} and \bibinfo{author}{L.~Hackl},
  \bibinfo{title}{Bosonic and fermionic gaussian states from k\"ahler
  structures}   (\bibinfo{year}{unpublished}).

\bibitem{Bianchi2017entropy}
\bibinfo{author}{E.~Bianchi} and \bibinfo{author}{A.~Satz},
  \bibinfo{title}{Entropy of a subalgebra of observables and the geometric
  entanglement entropy}   (\bibinfo{year}{unpublished}).

\bibitem{ghosh2017entanglement}
\bibinfo{author}{S.~Ghosh}, \bibinfo{author}{K.~S. Gupta}, and
  \bibinfo{author}{S.~C. Srivastava}, \bibinfo{title}{Entanglement dynamics
  following a sudden quench: an exact solution},
  \href{https://arxiv.org/abs/1709.02202}{\bibinfo{howpublished}{arXiv:1709.02202}}.

\bibitem{khlebnikov_kruczenski_14}
\bibinfo{author}{S.~Khlebnikov} and \bibinfo{author}{M.~Kruczenski},
  \bibinfo{title}{Locality, entanglement, and thermalization of isolated
  quantum systems},
  \bibinfo{journal}{\href{http://dx.doi.org/10.1103/PhysRevE.90.050101}{Phys.
  Rev. E}} \href{http://dx.doi.org/10.1103/PhysRevE.90.050101}{{\bf
  \bibinfo{volume}{90}}, \bibinfo{pages}{050101}}
  (\href{http://dx.doi.org/10.1103/PhysRevE.90.050101}{\bibinfo{year}{2014}}).

\bibitem{alba14}
\bibinfo{author}{V.~Alba} and \bibinfo{author}{F.~Heidrich-Meisner},
  \bibinfo{title}{Entanglement spreading after a geometric quench in quantum
  spin chains},
  \bibinfo{journal}{\href{http://dx.doi.org/10.1103/PhysRevB.90.075144}{Phys.
  Rev. B}} \href{http://dx.doi.org/10.1103/PhysRevB.90.075144}{{\bf
  \bibinfo{volume}{90}}, \bibinfo{pages}{075144}}
  (\href{http://dx.doi.org/10.1103/PhysRevB.90.075144}{\bibinfo{year}{2014}}).

\bibitem{alba_17}
\bibinfo{author}{V.~Alba}, \bibinfo{title}{Entanglement and quantum transport
  in integrable systems},
  \href{https://arxiv.org/abs/1706.00020}{\bibinfo{howpublished}{arXiv:1706.00020v1}}.

\bibitem{prosen.2008}
\bibinfo{author}{M.~\ifmmode \check{Z}\else
  \v{Z}\fi{}nidari\ifmmode~\check{c}\else \v{c}\fi{}},
  \bibinfo{author}{T.~Prosen}, and
  \bibinfo{author}{P.~Prelov\ifmmode~\check{s}\else \v{s}\fi{}ek},
  \bibinfo{title}{Many-body localization in the heisenberg {XXZ} magnet in a
  random field},
  \bibinfo{journal}{\href{http://dx.doi.org/10.1103/PhysRevB.77.064426}{Phys.
  Rev. B}} \href{http://dx.doi.org/10.1103/PhysRevB.77.064426}{{\bf
  \bibinfo{volume}{77}}, \bibinfo{pages}{064426}}
  (\href{http://dx.doi.org/10.1103/PhysRevB.77.064426}{\bibinfo{year}{2008}}).

\bibitem{bardarson.2012}
\bibinfo{author}{J.~H. Bardarson}, \bibinfo{author}{F.~Pollmann}, and
  \bibinfo{author}{J.~E. Moore}, \bibinfo{title}{Unbounded growth of
  entanglement in models of many-body localization},
  \bibinfo{journal}{\href{http://dx.doi.org/10.1103/PhysRevLett.109.017202}{Phys.
  Rev. Lett.}} \href{http://dx.doi.org/10.1103/PhysRevLett.109.017202}{{\bf
  \bibinfo{volume}{109}}, \bibinfo{pages}{017202}}
  (\href{http://dx.doi.org/10.1103/PhysRevLett.109.017202}{\bibinfo{year}{2012}}).

\bibitem{serbyn.2013}
\bibinfo{author}{M.~Serbyn}, \bibinfo{author}{Z.~Papi{\'c}}, and
  \bibinfo{author}{D.~A. Abanin}, \bibinfo{title}{Universal slow growth of
  entanglement in interacting strongly disordered systems},
  \bibinfo{journal}{\href{http://dx.doi.org/10.1103/PhysRevLett.110.260601}{Phys.
  Rev. Lett.}} \href{http://dx.doi.org/10.1103/PhysRevLett.110.260601}{{\bf
  \bibinfo{volume}{110}}, \bibinfo{pages}{260601}}
  (\href{http://dx.doi.org/10.1103/PhysRevLett.110.260601}{\bibinfo{year}{2013}}).

\bibitem{deutsch1991}
\bibinfo{author}{J.~M. Deutsch}, \bibinfo{title}{Quantum statistical mechanics
  in a closed system},
  \bibinfo{journal}{\href{http://dx.doi.org/10.1103/PhysRevA.43.2046}{Phys.
  Rev. A}} \href{http://dx.doi.org/10.1103/PhysRevA.43.2046}{{\bf
  \bibinfo{volume}{43}}, \bibinfo{pages}{2046}}
  (\href{http://dx.doi.org/10.1103/PhysRevA.43.2046}{\bibinfo{year}{1991}}).

\bibitem{srednicki1994}
\bibinfo{author}{M.~Srednicki}, \bibinfo{title}{Chaos and quantum
  thermalization},
  \bibinfo{journal}{\href{http://dx.doi.org/10.1103/PhysRevE.50.888}{Phys. Rev.
  E}} \href{http://dx.doi.org/10.1103/PhysRevE.50.888}{{\bf
  \bibinfo{volume}{50}}, \bibinfo{pages}{888}}
  (\href{http://dx.doi.org/10.1103/PhysRevE.50.888}{\bibinfo{year}{1994}}).

\bibitem{rigol.2008}
\bibinfo{author}{M.~Rigol}, \bibinfo{author}{V.~Dunjko}, and
  \bibinfo{author}{M.~Olshanii}, \bibinfo{title}{Thermalization and its
  mechanism for generic isolated quantum systems},
  \bibinfo{journal}{\href{http://dx.doi.org/10.1038/nature06838}{Nature}}
  \href{http://dx.doi.org/10.1038/nature06838}{{\bf \bibinfo{volume}{452}},
  \bibinfo{pages}{854}}
  (\href{http://dx.doi.org/10.1038/nature06838}{\bibinfo{year}{2008}}).

\bibitem{khatami_rigol_12a}
\bibinfo{author}{E.~Khatami}, \bibinfo{author}{M.~Rigol},
  \bibinfo{author}{A.~Rela\~no}, and \bibinfo{author}{A.~M.
  Garc\'{\i}a-Garc\'{\i}a}, \bibinfo{title}{Quantum quenches in disordered
  systems: {A}pproach to thermal equilibrium without a typical relaxation
  time},
  \bibinfo{journal}{\href{http://dx.doi.org/10.1103/PhysRevE.85.050102}{Phys.
  Rev. E}} \href{http://dx.doi.org/10.1103/PhysRevE.85.050102}{{\bf
  \bibinfo{volume}{85}}, \bibinfo{pages}{050102(R)}}
  (\href{http://dx.doi.org/10.1103/PhysRevE.85.050102}{\bibinfo{year}{2012}}).

\bibitem{serbyn_Papic_14}
\bibinfo{author}{M.~Serbyn}, \bibinfo{author}{Z.~Papi\ifmmode~\acute{c}\else
  \'{c}\fi{}}, and \bibinfo{author}{D.~A. Abanin}, \bibinfo{title}{Quantum
  quenches in the many-body localized phase},
  \bibinfo{journal}{\href{http://dx.doi.org/10.1103/PhysRevB.90.174302}{Phys.
  Rev. B}} \href{http://dx.doi.org/10.1103/PhysRevB.90.174302}{{\bf
  \bibinfo{volume}{90}}, \bibinfo{pages}{174302}}
  (\href{http://dx.doi.org/10.1103/PhysRevB.90.174302}{\bibinfo{year}{2014}}).

\bibitem{tang2015quantum}
\bibinfo{author}{B.~Tang}, \bibinfo{author}{D.~Iyer}, and
  \bibinfo{author}{M.~Rigol}, \bibinfo{title}{Quantum quenches and many-body
  localization in the thermodynamic limit},
  \bibinfo{journal}{\href{http://dx.doi.org/10.1103/PhysRevB.91.161109}{Phys.
  Rev. B}} \href{http://dx.doi.org/10.1103/PhysRevB.91.161109}{{\bf
  \bibinfo{volume}{91}}, \bibinfo{pages}{161109(R)}}
  (\href{http://dx.doi.org/10.1103/PhysRevB.91.161109}{\bibinfo{year}{2015}}).

\bibitem{anderson.1958}
\bibinfo{author}{P.~W. Anderson}, \bibinfo{title}{Absence of diffusion in
  certain random lattices},
  \bibinfo{journal}{\href{http://dx.doi.org/10.1103/PhysRev.109.1492}{Phys.
  Rev.}} \href{http://dx.doi.org/10.1103/PhysRev.109.1492}{{\bf
  \bibinfo{volume}{109}}, \bibinfo{pages}{1492}}
  (\href{http://dx.doi.org/10.1103/PhysRev.109.1492}{\bibinfo{year}{1958}}).

\bibitem{Adesso:2007tx}
\bibinfo{author}{G.~Adesso} and \bibinfo{author}{F.~Illuminati},
  \bibinfo{title}{{Entanglement in continuous variable systems: Recent advances
  and current perspectives}},
  \bibinfo{journal}{\href{http://dx.doi.org/10.1088/1751-8113/40/28/S01}{J.
  Phys.}} \href{http://dx.doi.org/10.1088/1751-8113/40/28/S01}{{\bf
  \bibinfo{volume}{A40}}, \bibinfo{pages}{7821}}
  (\href{http://dx.doi.org/10.1088/1751-8113/40/28/S01}{\bibinfo{year}{2007}}).

\bibitem{holevo1999capacity}
\bibinfo{author}{A.~S. Holevo}, \bibinfo{author}{M.~Sohma}, and
  \bibinfo{author}{O.~Hirota}, \bibinfo{title}{Capacity of quantum gaussian
  channels},
  \bibinfo{journal}{\href{http://dx.doi.org/10.1103/PhysRevA.59.1820}{Phys.
  Rev. A}} \href{http://dx.doi.org/10.1103/PhysRevA.59.1820}{{\bf
  \bibinfo{volume}{59}}, \bibinfo{pages}{1820}}
  (\href{http://dx.doi.org/10.1103/PhysRevA.59.1820}{\bibinfo{year}{1999}}).

\bibitem{arnold1990symplectic}
\bibinfo{author}{V.~Arnold} and \bibinfo{author}{A.~Givental},
  \bibinfo{title}{Symplectic geometry, encyclopedia of mathematical science
  vol. 4}  (\bibinfo{year}{1990}).

\bibitem{de2015multimode}
\bibinfo{author}{G.~De~Palma}, \bibinfo{author}{A.~Mari},
  \bibinfo{author}{S.~Lloyd}, and \bibinfo{author}{V.~Giovannetti},
  \bibinfo{title}{Multimode quantum entropy power inequality},
  \bibinfo{journal}{Physical Review A} {\bf \bibinfo{volume}{91}},
  \bibinfo{pages}{032320}  (\bibinfo{year}{2015}).

\bibitem{weintraub2009jordan}
\bibinfo{author}{S.~H. Weintraub}, \bibinfo{title}{Jordan canonical form:
  theory and practice}, \bibinfo{journal}{Synthesis Lectures on Mathematics and
  Statistics} {\bf \bibinfo{volume}{2}}, \bibinfo{pages}{1}
  (\bibinfo{year}{2009}).

\bibitem{Knuth.1997}
\bibinfo{author}{D.~E. Knuth}, {\em \bibinfo{title}{The Art of Computer
  Programming, Volume 1 (3rd Ed.): Fundamental Algorithms}\/}
  (\bibinfo{publisher}{Addison Wesley Longman Publishing Co., Inc.},
  \bibinfo{address}{Redwood City, CA, USA}, \bibinfo{year}{1997}).

\end{thebibliography}

\end{document}